\documentclass[cmp]{svjour}
\usepackage{amsmath}
\usepackage{amsfonts,amssymb}
\journalname{Communications in Mathematical Physics}

\newcommand{\fsp}[1]{\ll{#1}\gg_{\rm{free}}}

\usepackage{color}

\newcommand\bx{{\mathbf x}}
\newcommand\by{{\mathbf y}}

\newcommand\be{{\mathbf e}}
\newcommand\bj{{\mathbf j}}

\newcommand\bk{{\mathbf k}}
\newcommand\bu{{\mathbf u}}
\newcommand\bv{{\mathbf v}}

\newcommand\bG{{\mathbf G}}
\newcommand\bJ{{\mathbf J}}

\newcommand\LL{{\mathbb L}}

\newcommand\PP{{\mathbb P}}
\newcommand\NN{{\mathbb N}}
\newcommand\RR{{\mathbb R}}
\newcommand\TT{{\mathbb T}}
\newcommand\ZZ{{\mathbb Z}}

\newcommand\ve{\varepsilon}

\newcommand{\mc}[1]{{\mathcal #1}}
\newcommand{\mf}[1]{{\mathfrak #1}}
\newcommand{\mb}[1]{{\mathbf #1}}
\newcommand{\bb}[1]{{\mathbb #1}}

\newtheorem{prop}{Proposition}
\newtheorem{theo}{Theorem}

\usepackage{url}

\begin{document}

\title{{Anomalous fluctuations for a perturbed Hamiltonian system with exponential interactions}}
\titlerunning{Perturbed Hamiltonian system with exponential interactions}

\author{C\'edric Bernardin \inst{1} \and Patr\'icia Gon\c calves \inst{2}}
\institute{Universit\'e de Lyon and CNRS, UMPA, UMR-CNRS 5669, ENS-Lyon,
46, all\'ee d'Italie, 69364 Lyon Cedex 07 - France. \\ \email{Cedric.Bernardin@ens-lyon.fr}  \and
  CMAT, Centro de Matem\'atica da Universidade do Minho, Campus de Gualtar, 4710-057 Braga, Portugal. \\ \email{patg@math.uminho.pt} }
\authorrunning{C. Bernardin and P. Gon\c calves}

\maketitle
\begin{abstract}
A one-dimensional Hamiltonian system with exponential interactions perturbed by a conservative noise is considered. It is proved that energy superdiffuses and upper and lower bounds describing this anomalous diffusion are obtained.
\end{abstract}
\section{Introduction}
Over the last decade, transport properties of one-dimensional Hamiltonian systems consisting
of coupled oscillators on a lattice have been the subject of many theoretical and
numerical studies, see the review papers~\cite{BLR,D,LLP}.
Despite many efforts, our knowledge of the fundamental mechanisms necessary and/or sufficient
to have a normal diffusion remains very limited.
Nevertheless, it has been recognized that conservation of momentum plays a major
role and numerical simulations provide a strong evidence of the fact that
one dimensional chains of anharmonic oscillators conserving momentum are usually  {\footnote{See however the coupled-rotor model which displays normal behavior (see \cite{LLP}, Section 6.4). }} superdiffusive.

An interesting area of current research consists in studying this problem for hybrid models where a stochastic perturbation is superposed to the deterministic evolution. Even if the problem is considerably simplified, several open challenging questions can be addressed for these systems.  In \cite{BBO2} it is  proved that the thermal conductivity of an unpinned harmonic chain of oscillators perturbed by an energy-momentum conservative noise is infinite while if a pinning potential (destroying momentum conservation) is added it is finite. In the same paper, diverging upper bounds are provided when some nonlinearities are added. This does not, however, exclude the possibility of having a finite conductivity. Therefore much more interesting would be to obtain lower bounds showing that the conductivity is infinite and that energy superdiffuses, but this problem is left open in \cite{BBO2}.

In \cite{BS}, has been introduced and studied numerically, a class of Hamiltonian models for which anomalous diffusion is observed. There, the investigated systems present strong analogies with standard chains of oscillators. They can be described as follows. Let $V$ and $U$ be two non-negative potentials on $\RR$ and consider the Hamiltonian system $( \, {\bf r} (t) , {\bf p} (t) \,)_{t \ge 0}$ whose equations of motion are given by
\begin{equation}
\label{eq:generaldynamics}
\frac{dp_x}{dt} = V'(r_{x+1}) -V'(r_x),
\qquad \frac{dr_x}{dt} = U' (p_x) -U' (p_{x-1}),
\qquad x \in \ZZ,
\end{equation}
where $p_x$ is the momentum of the particle $x$, $q_x$ its position and $r_x=q_{x} -q_{x-1}$ is the ``deformation'' of the lattice at $x$. Standard chains of oscillators are recovered for a quadratic kinetic energy $U(p)=p^2 /2$. Now, take $V=U$, and call $\eta_{2x-1}=r_x$ and $\eta_{2x}=p_x$. The dynamics can be rewritten as:
\begin{equation}
\label{eq:dyneq}
d\eta_{x} (t) =\Big(V' (\eta_{x+1}) - V' (\eta_{x-1})\Big) dt.
\end{equation}
Notice that with these new variables the energy of the system is simply given by $\sum_{x\in \bb Z} V(\eta_x)$. In \cite{BS} an anomalous diffusion of energy is numerically observed for a generic potential $V$. Then, following the spirit of \cite{BBO2}, the deterministic evolution is perturbed by adding a noise which consists to exchange $\eta_{x}$ with $\eta_{x+1}$ at random exponential times, independently for each bond $\{x,x+1\}$. The dynamics still conserves the energy $\sum_{x\in \ZZ} V(\eta_x)$ and the ``volume'' $\sum_{x\in \ZZ} \eta_x$ and destroys all other conserved quantities. As argued in \cite{BS}, the volume conservation law is responsible for the anomalous energy diffusion observed for this class of energy-volume conserving dynamics. This can be shown for quadratic interactions (\cite{BS}) with a behavior similar to the one observed in \cite{BBO2}. For nonlinear interactions the problem is much more difficult.

The aim of this paper is to show that if the interacting potential is of exponential type then the energy superdiffuses. Therefore, for this class of related models, in a particular case, we answer to the open question stated in \cite{BBO2}. With some additional technical work we think that our methods could be carried out to the Toda lattice perturbed by an energy-momentum conserving noise (considered e.g. in \cite{ILOS}). The exponential form of the potential $V$ makes the deterministic dynamics given by (\ref{eq:dyneq}) completely integrable. Nevertheless our proofs do not rely on this exceptional property of the dynamics and could be potentially generalized to other potentials $V$. The main ingredient used is the existence of explicit orthogonal polynomials for the equilibrium measures (see Section \ref{sec:dua}).

The paper is organized as follows. In Section \ref{sec:model} we define precisely the model. The results are stated in Section \ref{sec:results}. To prove the theorems we first perform a microscopic change of variables (Section \ref{sec:cv}) which permits to use a nice orthogonal decomposition of the generator (Section \ref{sec:dua}). Roughly speaking the upper bound on the energy superdiffusion is proved in Section \ref{sec:triv} and  the lower bound in Section \ref{sec:diff}. Section \ref{sec:pert} contains a comment about the possible extensions and comparisons of our model to others. In the Appendix we prove the existence of the infinite dynamics.

\vspace{0.2cm}

{\textbf{Notations}:  For any $a, b \in \RR^2$, $a \cdot b $ stands for the standard scalar product between $a$ and $b$ and $|a|= \sqrt{a \cdot a}$ for the norm of $a$. The transpose of a matrix $M$ is denoted by $M^T$. If $u: \bx=(x_1, \ldots,x_n)^T \in \RR^n \to u(\bx) =(u_1(\bx), \ldots, u_d (\bx))^T \in \RR^d$ is a differentiable function then $\partial_{x_j} u_i (\bx)$ denotes the partial derivative of $u_j$ with respect to the $j$-th coordinate at $\bx$ and $\nabla u( \bx)$ denotes the differential matrix (the gradient if $d=1$) of $u$ at $\bx$, i.e. the $n \times d$ matrix whose $(i,j)$-th entry is $\partial_{x_j} u_i(\bx)$; if $u:=(u_1, \ldots, u_d)^T:\ZZ \to \RR^d$ then we adopt the same notation to denote the discrete gradient of $u$ defined by $\nabla u:= (\nabla u_1, \ldots,\nabla u_d)^T$ with $\nabla u_i (x)= u_i (x+1) -u_i (x)$.
}

\vspace{0.2cm}

\section{The model}
\label{sec:model}
Let $b>0$ and $V_{b} (q)= e^{-bq} -1+bq $. We consider the system $\eta(t)=\{\eta_x(t):x\in{\mathbb{Z}}\}$ on $\RR^{\ZZ}$ defined by its generator $L=A+\gamma S$, $\gamma>0$, where for local {\footnote{A function $f$ defined on an infinite product space is said to be local if it depends only on its variable through a finite number of coordinates.}} differentiable functions $f:\RR^{\ZZ}\rightarrow{\mathbb{R}}$ we have that
\begin{equation*}
(Af)(\eta)=\sum_{x \in \ZZ} \Big(V_b^{\prime} (\eta_{x+1}) -V_b^{\prime} (\eta_{x-1})  \Big) (\partial_{\eta_x} f)(\eta)
\end{equation*}
and
\begin{equation*}
 (Sf)(\eta)=\sum_{x \in \ZZ} \Big( f(\eta^{x,x+1}) -f(\eta) \Big),
\end{equation*}
 where $\eta^{x,x+1}$ is obtained from $\eta$ by exchanging the variables $\eta_x$ and $\eta_{x+1}$, namely
\begin{equation}\label{etax,x+1}
\eta^{x,x+1}_y=\left\{\begin{array}{cl}
\eta_{x+1},& \mbox{if}\,\,\, y=x\,,\\
\eta_x,& \mbox{if} \,\,\,y=x+1\,,\\
\eta_y,& \mbox{otherwise}\,.
\end{array}
\right.
\end{equation}
The deterministic system (\ref{eq:dyneq}) with potential $V_b$ is well known in the integrable systems literature. It has been introduced in \cite{KVM} by Kac and van Moerbecke and was shown to be completely integrable. Consequently, the energy transport is ballistic (\cite{BS,Z}). As we will see this is different when the noise is added: the energy transport is no more ballistic but superdiffusive.

The existence of the dynamics generated by $L$ is proved in the Appendix for a large set of initial conditions and in particular for a set of full measure w.r.t. any invariant state $\mu_{{\bar \beta},{\bar \lambda}}$ (see bellow for its definition).

The system conserves the energy $\sum_{x \in \ZZ} V_{b} (\eta_x)$ and the volume $\sum_{x \in \ZZ} \eta_x$. In fact, we have
\begin{equation*}
{L} (V_{b} (\eta_x))=-\nabla {\bar j}_{x-1,x}(\eta), \quad L (\eta_x)=-\nabla {\bar j}^{\prime}_{x-1,x}(\eta),
\end{equation*}
where the microscopic currents are given by
\begin{equation*}
{\bar j}_{x,x+1}(\eta)=-b^2 e^{-b(\eta_x + \eta_{x+1})}+b^2(e^{-b \eta_x} +e^{-b \eta_{x+1}})- \gamma \nabla V_{b} (\eta_{x})
\end{equation*}
and
\begin{equation*}
{\bar j}^{\prime}_{x,x+1}(\eta)=b e^{-b \eta_x} +b e^{-b\eta_{x+1}} -\gamma \nabla \eta_x .\\
\end{equation*}

Every product probability measure $\mu_{\bar \beta, \bar \lambda}$ {{on $\bb R^{\bb Z}$}} in the form
$$\mu_{{\bar \beta}, {\bar \lambda}} (d\eta) = \prod_{x\in \bb Z} {\bar Z}^{-1} ({\bar \beta},{\bar \lambda}) \exp \{ -{\bar \beta} e^{-b \eta_x} - {\bar\lambda} \eta_x\}d\eta_x, \quad {\bar \beta}>0\, , \, {\bar\lambda} >0$$
is invariant for the dynamics.

Let $\langle \cdot\rangle_{\mu_{\bar \beta, \bar \lambda}}$ denote the average with respect to $\mu_{\bar \beta, \bar \lambda}$.
We define ${\bar e}:=\bar {e} ({\bar \beta}, {\bar \lambda}), {\bar v}:=\bar{v} ({\bar \beta}, {\bar \lambda})$ as the averages of the conserved quantities $V_b (\eta_x)$, $\eta_x$ with respect to $\mu_{\bar \beta, \bar \lambda}$, respectively, namely ${\bar e}=\langle V_b(\eta_x) \rangle_{\mu_{\bar \beta, \bar \lambda}}$ and
${\bar v}=\langle \eta_x \rangle_{\mu_{\bar \beta, \bar \lambda}}$.

A simple computation shows that
 \begin{equation}\label{mean of micro currents}
 \langle {\bar j}_{x,x+1} \rangle_{\mu_{\bar \beta, \bar \lambda}}=-b^2({\bar e} -b {\bar v})^2+b^2 \quad \textrm{ and} \quad \langle {\bar j}^{\prime}_{x,x+1}  \rangle_{\mu_{\bar \beta, \bar \lambda}}  =2b ({\bar e} -b {\bar v}+1).
\end{equation}
Hence, in the hyperbolic scaling, the hydrodynamical equations  are given by
\begin{equation}
\label{eq:hl1euler}
\begin{cases}
\partial_t {\mf e} -b^2 \, \partial_q ( ({\mf e} - b{\mf v})^2) =0\\
\partial_{t} {\mf v} + 2b \, \partial_q ({\mf e -b {\mf v}}) =0
\end{cases}
\end{equation}
and can be written in the compact form $\partial_t {\bar {\mf X}} + \partial_q {\bar{ \mf J}} ({\bar {\mf X}}) =0$ with
\begin{equation}\label{eq:hl-ss00}
\bar {\mf X}=
\left(
\begin{array}{c}
{\mf e}\\
{\mf v}
\end{array}
\right)
, \quad \textrm{and} \quad
\bar{ \mf J }({\bar {\mf X}})=
\left(
\begin{array}{c}
-b^2 ( {\mf e} - b{\mf v})^2\\
2b  ( {\mf e} - b{\mf v})
\end{array}
\right).
\end{equation}
This can be proved before the appearance of the shocks (see \cite{BS}).  The differential matrix of $\bar{ \mf J}$ is given by
\begin{equation*}
\nabla \bar{ \mf J}(\bar {\mf X})=
2b \left(
\begin{array}{cc}
-b ( {\mf e} - b{\mf v})& b^2 ( {\mf e} - b{\mf v}) \\
1 & -b
\end{array}
\right).
\end{equation*}
For given $({\bar e}, {\bar v})$ we denote by $({\bar T}^{+}_t )_{t \ge 0}$ (resp. $({\bar T}^{-}_t)_{t \ge 0}$) the semigroup on $S(\RR) \times S(\RR)$ generated by
\begin{equation}\label{linearized system 1}
\partial_t \varepsilon + {\bar M}^T\, \partial_q \varepsilon =0, \quad ({\text{resp.}} \; \partial_t \varepsilon - {\bar M}^T \, \partial_q \varepsilon =0 ),
\end{equation}
where
\begin{equation*}
{\bar M} := {\bar M} ({\bar e}, {\bar v})= [\nabla \bar{ \mf J}] (\bar \omega), \quad \bar \omega= \left(
\begin{array}{c}
{\bar e}\\
{\bar v}
\end{array}
\right).
\end{equation*}
We omit the dependence of these semigroups on $(\bar e , \bar v)$ for lightness of the notations. Above $S(\RR)$ denotes the Schwartz space of smooth
rapidly decreasing functions.

\section{Statement of the results}
\label{sec:results}

For each integer $z \ge 0$, let $H_z (x) = (-1)^z e^{x^2} \cfrac{d^z}{dx^z} e^{-x^2}$ be the Hermite polynomial and $h_z (x) =(z! {\sqrt{2\pi}})^{-1} H_{z} (x) e^{-x^2}$ the Hermite function. The set $\{ h_z, z\ge 0\}$ is an orthonormal basis of ${\bb L}^2 (\RR)$. Consider in ${\bb L}^2(\RR)$ the operator $K_0 = x^2-\Delta$, $\Delta$ being the Laplacian on $\RR$. For an integer $k \ge 0$, denote by ${\bb H}_k$ the Hilbert space induced by $S (\RR)$ and the scalar product $\langle \cdot,\cdot\rangle_{k}$ defined by $\langle f, g \rangle_k= \langle f, K_0^k g \rangle_0$, where $\langle \cdot,\cdot\rangle_0$ denotes the inner product of $\LL^2 (\RR)$ and denote by ${\bb H}_{-k}$ the dual of ${\bb H}_k$, relatively to this inner product. Let $\langle\cdot\rangle$ represent the average with respect to the Lebesgue measure.

We take the infinite system at equilibrium under the Gibbs measure $\mu_{\bar \beta,\bar \lambda}$ corresponding to a mean energy $\bar e$ and a mean volume $\bar v$. Our goal is to study
the energy-volume fluctuation field in the time-scale $tn^{1+\alpha}$, $\alpha \ge0$:
\begin{equation}
  \label{eq:YY}
  \mathcal{Y}^{n,\alpha}_t  (\bG) =\frac{1}{\sqrt{n}} \sum_{x\in \ZZ}
  \bG\left(x/n\right) \cdot \left({\bar \omega}_x (tn^{1+\alpha})  - {\bar \omega}\right),
\end{equation}
where for $q\in{\mathbb{R}}$,  $x \in \ZZ$,
\[
\bG(q) = \left(
\begin{array}{c}
G_{1} (q)\\
G_2 (q)
\end{array}
\right), \quad
{\bar \omega}_x= \left(
\begin{array}{c}
V_b (\eta_x) \\
\eta_x
\end{array}
\right)
\]
and $G_1, G_2$ are test functions belonging to $S(\RR)$.

If $E$ is a Polish space then $D(\RR^+, E)$ (resp. $C(\RR^+, E)$) denotes the space of $E$-valued functions, right continuous with left limits (resp. continuous), endowed with the Skorohod (resp. uniform) topology. Let $Q^{n,\alpha}$ be the probability measure on ${D}(\RR^+,{\bb H}_{-k} \times {\bb H}_{-k})$ induced by the fluctuation field ${\mc Y}^{n,\alpha}_t$ and $\mu_{\bar \beta,\bar \lambda}$. Let $\mathbb{P}_{\mu_{\bar \beta,\bar \lambda}}$ denote the probability measure on ${ D}(\RR^+, \RR^{\ZZ})$ induced by $(\eta(t))_{t\geq{0}}$ and $\mu_{\bar \beta,\bar \lambda}$. Let $\mathbb{E}_{\mu_{\bar \beta,\bar \lambda}}$ denote the expectation with respect to $\mathbb{P}_{\mu_{\bar \beta,\bar \lambda}}$.

\begin{theo}
\label{th:fluct-hs}
Fix an integer $k>2$. Denote by $Q$ the probability measure on $C(\RR^+, {\bb H}_{-k} \times {\bb H}_{-k})$ corresponding to a stationary Gaussian process with mean $0$ and covariance given by
\begin{equation*}
{\mathbb E}_{Q} \left[ \mathcal{Y}_t (\mb H) \, \mathcal{ Y}_s (\mb G) \right] =  \langle \,{\bar T}_t^{-} \mb H\;  \cdot \;   \bar \chi \;  {\bar T}_s^{-}\mb G \,  \rangle
\end{equation*}
for every $0 \le s \le t$ and $\mb H, \mb G$ in ${\bb H}_k \times {\bb H}_k$. Here ${\bar \chi}:={\bar \chi} ({\bar \beta}, {\bar \lambda})$ is the equilibrium covariance matrix {\footnote{See (\ref{eq:chibar2}) for an explicit expression.}} of ${\bar \omega}_0$. Then, the sequence $(Q^{n,0})_{n \ge 1}$ converges weakly,   as $n\to\infty$, to the probability measure $Q$.
\end{theo}

A byproduct of Theorem \ref{th:fluct-hs} is a Central Limit Theorem for the energy flux and for the volume flux through a fixed bond. Despite it is not directly related to the problem of anomalous diffusion it has a probabilistic interest.  For that purpose, fix a site $x\in{\bb Z}$, let $\mathcal{E}_{x,x+1}^n(t)$ (resp. $\mathcal{V}_{x,x+1}^n(t)$) denote the energy (resp. volume) flux through the bond $\{x,x+1\}$ during the time interval $[0,tn]$. By conservation laws, for any $x\in{\mathbb{Z}}$ it holds that:
\begin{equation*}
\mathcal{E}_{x-1,x}^n(t)-\mathcal{E}_{x,x+1}^n(t):=V_b(\eta_x(tn))-V_b(\eta_x(0))
\end{equation*}
\begin{equation*}
 \Big( \text{resp. \; } \mathcal{V}^n_{x-1,x}(t)-\mathcal{V}_{x,x+1}^n(t):=\eta_x(tn)-\eta_x(0)\Big).
\end{equation*}
 This, together with the previous result allow us to conclude that
\begin{corollary}\label{CLT Energy flux}
 Fix $x\in\bb Z$ and let $Z^{n,e}_t:=\frac{1}{\sqrt n}\{\mathcal{E}_{x,x+1}^n(t)- \bb E_{\mu_{\bar \beta,\bar \lambda}}[\mathcal{E}_{x,x+1}^n(t)]\}$. For every $t\geq{0}$, $(Z^{n,e}_{t})_{n\geq 1}$ converges in law in the sense of finite-dimensional distributions, as $n\to\infty$, to a Brownian motion $Z_{t}^e$ with mean zero and covariance given by
\begin{equation*}
{\mathbb E}_{Q} [Z^e_tZ^e_s]=\frac{2}{\bar\beta^3}(\bar\lambda-b\bar\beta)^2s,
\end{equation*}
for all $s\leq{t}$.
\end{corollary}
\begin{corollary}\label{CLT volume flux}
 Fix $x\in\bb Z$ and let $Z^{n,v}_t:=\frac{1}{\sqrt n}\{\mathcal{V}_{x,x+1}^n(t)-{{ \bb E_{\mu_{\bar \beta,\bar \lambda}}}}[\mathcal{V}_{x,x+1}^n(t)]\}$. For every $t\geq 0$, $(Z^{n,v}_t)_{n\geq 1}$ converges in law in the sense of finite-dimensional distributions, as $n\to\infty$, to a Brownian motion  $Z_{t}^v$ with mean zero and covariance given by
\begin{equation*}
{\mathbb E}_{Q}[Z^v_tZ^v_s]=\frac{2}{\bar\beta}s,
\end{equation*}
for all $s\leq{t}$.
\end{corollary}

We notice that, according to Corollary \ref{CLT Energy flux}, the limiting {{energy flux}} $Z_t^e$ has a vanishing variance for $\bar\lambda=b\bar\beta$ which is equivalent to $\bar{e} = b\bar{v}$. Last equivalence  is  a consequence of \eqref{eq:chimpot} and \eqref{eq:relblbl}.

%

The theorem above means that in the hyperbolic scaling the fluctuations are trivial: the initial fluctuations are transported by the linearized system of (\ref{eq:hl1euler}). To see a nontrivial behavior we have to study, in the transport frame, the fluctuations at a longer time scale $t n^{1+\alpha}$, with $\alpha>0$. Thus, we consider the fluctuation field ${\widehat {\mc Y}}_{\cdot}^{n,\alpha}$, $\alpha>0$, defined, for any $\bG \in S(\RR) \times S(\RR)$, by
\begin{equation}\label{longer density field}
{\widehat {\mc Y}}_t^{n,\alpha} (\bG)= {\mc Y}_t^{n, \alpha} \left( {\bar T}^+_{t n^{\alpha}} \bG \right).
\end{equation}
According to the fluctuating hydrodynamics theory (\cite{Sp}, pp. 85-96), in the case of a normal (diffusive) behavior $\alpha=1$,  the field $( {\widehat{\mc Y}}_t^{n,\alpha} )\,_{t \ge 0}$ should converge to the stationary field $({\widehat{\mc Y}}_t \,)\,_{t \ge 0}$ simply related to the solution $(\widehat{\mc Z}_t\,)\,_{t \ge 0}$ of the linear two dimensional vector valued (infinite-dimensional) stochastic partial differential equation
\begin{equation}
\label{eq:OUeq}
\partial_t {\widehat {\mc Z_t}}= \nabla \cdot \left( \, {\mc D} \,  \nabla {\widehat {\mc Z}_t} \, \, \right) + \sqrt{2  {\mc D} {\bar \chi}}\,  \nabla \cdot W{{_t}}.
\end{equation}
Here $W_t$ is a standard two-dimensional vector valued space-time white noise and the coefficient ${\mc D}:= {\mc D} ({\bar e},{\bar v})$ is expressed by a Green-Kubo formula ( see (\ref{eq:GK00})).
As above, let $\widehat Q^{n,\alpha}$ be the probability measure on ${D}(\RR^+,{\bb H}_{-k} \times {\bb H}_{-k})$ induced by the fluctuation field $\widehat{\mc Y}^{n,\alpha}_t$ and $\mu_{\bar \beta,\bar \lambda}$.
Our second main theorem shows that the correct scaling exponent $\alpha$ is greater or equal than $1/3$:

\begin{theo}
\label{th:fluct-ds}
Fix an integer $k>1$ and $\alpha<1/3$.  Denote by $Q$ the probability measure on $C(\RR^+, {\bb H}_{-k} \times {\bb H}_{-k})$ corresponding to a stationary Gaussian process with mean $0$ and covariance given by
\begin{equation*}
{\mathbb E}_{Q} \left[ \mathcal{Y}_t (\mb H) \, \mathcal{ Y}_s (\mb G) \right] =  \langle \,  \mb H \;\cdot    \;  {\bar \chi} \;  \mb G   \rangle
\end{equation*}
for every $0 \le s \le t$ and $\mb H, \mb G$ in ${\bb H}_k \times {\bb H}_k$. Then, the sequence $(\widehat Q^{n,\alpha})_{n \ge 1}$ converges weakly, as $n\to\infty$, to the probability measure $Q$.
\end{theo}

As in the hyperbolic time scale from the previous result we obtain limiting results for the energy flux and volume flux. In this case,  we need to define the energy and volume flux through the time dependent bond ${\{u_t^{x,\alpha}(n),u_t^{x,\alpha}(n)+1\}}$, where $u_t^{x,\alpha}(n):= \lfloor x-\frac{-2b\bar\lambda}{\bar \beta} tn^{1+\alpha}\rfloor$ and  $\lfloor u\rfloor$ denotes the biggest integer number smaller or equal to $u$. The justification for taking this reference frame with precisely this velocity will be given ahead in Remark \ref{velocity}. Now, fix a site $x\in{\mathbb{Z}}$ and let $\mathcal{E}^n_{u_t^{x,\alpha}(n)}$ (resp. $\mathcal{V}_{u_t^{x,\alpha}(n)}^n(t)$) denote the energy (resp. volume) flux through the bond $\{u_t^{x,\alpha}(n),u_t^{x,\alpha}(n)+1\}$ during the time interval $[0,tn^{1+\alpha}]$.
Then, from the previous result we conclude that
\begin{corollary}\label{vanishing of Energy flux}
Fix $t\geq{0}$, $x\in\bb Z$ and $\alpha<1/3$. Then
\begin{equation*}
\lim_{n \to \infty}{{ \bb E_{\mu_{\bar \beta,\bar \lambda}}}}\left[ \left( \frac{1}{\sqrt n}\Big\{\mathcal{E}_{u_t^{x,\alpha}(n)}^n(t)-{{ \bb E_{\mu_{\bar \beta,\bar \lambda}}}}[\mathcal{E}_{u_t^{x,\alpha}(n)}^n(t)]\Big\}\right)^2\right]=0.
\end{equation*}
and
\begin{equation*}
\lim_{n \to \infty} {{ \bb E_{\mu_{\bar \beta,\bar \lambda}}}} \left[ \left( \frac{1}{\sqrt n}\Big\{\mathcal{V}_{u_t^{x,\alpha}(n)}^n(t)-{{ \bb E_{\mu_{\bar \beta,\bar \lambda}}}}[\mathcal{V}_{u_t^{x,\alpha}(n)}^n(t)]\Big\}\right)^2\right]=0.
\end{equation*}
\end{corollary}

Similar results have been obtained in \cite{G} by one of the authors for the asymmetric simple exclusion. The proof of Corollaries \ref{CLT Energy flux}, \ref{CLT volume flux} and \ref{vanishing of Energy flux} follows the same arguments as in \cite{G} once the previous theorems are proved. For that reason we will only give a sketch of their proof. The proof of the theorems is more problematic since the multi-scale analysis performed in \cite{G} relies crucially on the existence of a spectral gap so that we cannot follow \cite{G}. Therefore we propose an alternative approach based on computations of some resolvent norms.

Theorem \ref{th:fluct-ds} does not exclude the possibility of normal fluctuations, i.e. $\alpha=1$. In order to show that the system we consider is really superdiffusive we will show that the transport coefficient ${\mc D}$ which appears in (\ref{eq:OUeq})  is infinite so that the correct scaling exponent $\alpha$ is strictly smaller than $1$. Our third result, stated bellow, shows it is in fact less than $3/4$.

With the notations introduced in the previous section, the {\textit{normalized}} currents are defined by
\begin{equation}\label{norm. currents}
{\hat J}_{x,x+1}(\eta) =
\left(
\begin{array}{c}
{\bar j}_{x,x+1}(\eta)\\
{\bar j}^{\prime}_{x,x+1}(\eta)
\end{array}
\right)
- {\bar {\mf J}} ({\bar \omega}) - (\nabla {\bar{\mf J}}) ({\bar \omega})
\left(
\begin{array}{c}
V_b (\eta_x) -{\bar e}\\
\eta_x -{\bar v}
\end{array}
\right).
\end{equation}

Up to a constant matrix coming from a martingale term (due to the noise) and thus irrelevant for us (see \cite{BBO2}, \cite{BS}), the coefficient ${\mc D}$ is defined by the Green-Kubo formula
\begin{equation}
\label{eq:GK00}
{\mc D} = \int_{0}^{\infty} C(t) \, dt,
\end{equation}
where $$C(t):={\mathbb E}_{\mu_{\bar \beta, \bar \lambda}} \left[ \sum_{x \in \ZZ} {\hat J}_{x,x+1} (\eta(t)) \left[ {\hat J}_{0,1} (\eta(0)) \right]^T \right]$$
is the current-current correlation function.
The signature of the superdiffusive behavior of the system is seen in the divergence of the integral defining ${\mc D}$, i.e. in a slow decay of the current-current correlation function.  We introduce the Laplace transform function ${\mc F} (\gamma, \cdot)$ of the current-current correlation function. It is defined, for any $z>0$ by
\begin{equation*}
{\mc F} (\gamma, z) = \int_{0}^{\infty} e^{-z t} \, C(t)\, dt.
\end{equation*}

Our third theorem is the following lower bound on ${\mc F} (\gamma, z)$. Observe that ${\mc F} (\gamma, z)$ is a square matrix of size $2$ whose $(i,j)$-th entry  is denoted by ${\mc F}_{i,j}$.
\begin{theo}
\label{th:diffusivity}
Fix $\gamma>0$. For any $(i,j)\ne (1,1)$ and any $z>0$ we have
\begin{equation*}
{\mc F}_{i,j} (\gamma, z)=0.
\end{equation*}
There exists a positive constant $c:=c(\gamma)>0$ such that for any $z>0$,
\begin{equation*}
{\mc F}_{1,1} (\gamma, z) \ge c z^{-1/4}.
\end{equation*}
Moreover, there exists a positive constant $C:=C(\gamma)$ such that for any $z>0$,
\begin{equation}
\label{eq:F11c}
C^{-1} {\mc F}_{1,1} (1,z/\gamma) \le {\mc F}_{1,1} (\gamma, z) \le C {\mc F}_{1,1} (1,z/\gamma).
\end{equation}
\end{theo}

The lower bound ${\mc F}_{1,1} (\gamma, z) \ge c z^{1/4}$ means roughly that the current-current correlation function $C(t)$ is bounded by bellow by a constant times $t^{-3/4}$. The last part of the theorem is easy to prove but has an important consequence. In \cite{BS} numerical simulations are performed to detect the anomalous diffusion of energy. Since it is difficult to estimate numerically the time autocorrelation functions of the currents because of their expected long-time tails, a more tenable approach consists in studying a non equilibrium system in its steady state, i.e. considering a finite system in contact with two thermostats which fix the value of the energy at the boundaries. Then we estimate the dependence of the energy transport coefficient $\kappa (N)$ with the system size $N$. The latter is defined as $N$ times the average energy current. It turns out that $\kappa (N) \sim N^{\delta}$ with a parameter $\delta:=\delta(\gamma)>0$ increasing with the noise intensity $\gamma$ (except for the singular value $\delta=1$ when $\gamma=0$ which is a manifestation of the ballistic behavior of the Kac-van Moerbecke system). This result is very surprising since the more stochasticity in the model is introduced, the less the system is diffusive. The same has been observed for other anharmonic potentials in \cite{BS} and also for the Toda lattice perturbed by an energy-momentum conservative noise (\cite{ILOS}). It has been argued in \cite{ILOS} that this may be explained by the fact that some diffusive phenomena due to non-linearities, like localized breathers, are destroyed by the noise. In \cite{BDLLO} simulations have been performed directly with the Green-Kubo formula for other standard anharmonic chains with the same conclusion: current-current correlation function decreases slower when the noise intensity increases. If all these numerical simulations reproduce correctly the real behavior of the models investigated, they dismiss the theories which pretend that some universality holds, e.g. \cite{VB}. It is therefore very important to decide if the phenomena numerically observed are correct or not.

 Assuming that the current-current correlation function $C(t)$ has the time decay $C(t) \sim_{t \to \infty} t^{- \delta' (\gamma)}$, the inequality (\ref{eq:F11c}) shows that the exponent $\delta':=\delta' (\gamma)$ is independent of $\gamma$ (up to possible slowly varying functions corrections, i.e. in a Tauberian sense). It is usually argued but not proved (see e.g. the end of Section 5.3 in \cite{LLP}) that the exponent $\delta$ defined by the non-equilibrium stationary state is related to $\delta'$ by the relation $\delta= 1-\delta'$, and is consequently independent of $\gamma$ too. Therefore the numerical simulations do not seem to reflect the correct behavior of the system {\footnote{It would be very interesting to understand why the numerical simulations are so sensitive to the noise.}}.  A possible explanation of the inconsistency between the numerical observation and our result is simply that the relation $\delta=1-\delta'$ is not satisfied. Nevertheless, notice that the last part of our theorem is in fact valid for all the models cited above. It applies in particular to the models studied in \cite{BDLLO} and shows that the numerical observations of {that} paper, which are performed for the Green-Kubo formula, are not consistent with the real behavior of the system.

\section{A change of variables}
\label{sec:cv}

To study the energy-volume fluctuation field ${\mc Y}_{\cdot}^{n, \alpha}$, we introduce the following change of variables $\xi_x = e^{-b \eta_x}$, {for each $x\in \bb Z$}. Then, the previous Markovian system $(\eta (t))_{t \ge 0}$ defines a new Markovian system $(\xi (t))_{t \ge 0}$ with state space $(0,+\infty)^{\ZZ}$ whose generator ${\mc L}$ is equal to $b^2 {\mc A} +{\gamma} {\mc S}$, where for local differentiable functions $f:(0,+\infty)^{\ZZ}\rightarrow{\mathbb{R}}$ we have that
\begin{equation*}
({\mc A} f)(\xi)=\sum_{x \in \ZZ} \xi_x \Big( \xi_{x+1}- \xi_{x-1} \Big) (\partial_{\xi_x} f)(\xi)
\end{equation*}
and
 \begin{equation*}
 ({\mc S} f)(\xi)=\sum_{x \in \ZZ} \Big( f(\xi^{x,x+1}) -f(\xi) \Big),
\end{equation*}
 where $\xi^{x,x+1}$ is defined as in \eqref{etax,x+1}.

Observe that the energy and volume conservation laws correspond, for the process $(\xi (t))_{t \ge 0}$, to the conservation of the two following quantities
$\sum_{x \in \ZZ} {\xi_x}$ and $\sum_{x \in \ZZ} \log (\xi_x).$
The corresponding microscopic currents are defined by the conservation law equations:
\begin{equation*}
{\mc L} (\xi_x) = -\nabla j_{x-1,x}(\xi), \quad {\mc L} (\log \xi_x) = -\nabla j^{\prime}_{x-1,x}(\xi),
\end{equation*}
where
\begin{equation*}
j_{x,x+1}(\xi) = -b^2  \xi_{x} \xi_{x+1} - \gamma \nabla \xi_x,
\end{equation*}
and
\begin{equation*}
 j^{\prime}_{x,x+1}(\xi)= -b^2 (\xi_{x}+ \xi_{x+1})- \gamma \nabla \log (\xi_x).
\end{equation*}
We will use the compact notation
\begin{equation}\label{flux for xi}
J_{x,x+1}(\xi) =
\left(
\begin{array}{c}
j_{x,x+1} (\xi)\\
j^{\prime}_{x,x+1}(\xi)
\end{array}
\right).
\end{equation}

 Since $V_b(\eta_x)=\xi_x-\log( \xi_x)+1$ and $\eta_x=-\frac{1}{b}\log(\xi_x)$, we have the following relations between the microscopic currents
\begin{equation}
\label{eq:corr-jj}
{\bar j}_{x,x+1}(\eta)=j_{x,x+1}(\xi)-j^{\prime}_{x,x+1}(\xi), \quad\textrm{and}\quad {\bar j}_{x,x+1}^{\prime}(\eta)= -\cfrac{1}{b} j_{x,x+1}^{\prime}(\xi).
\end{equation}


%

If $\eta$ is distributed according to $\mu_{\bar \beta, \bar \lambda}$ then $\xi$ defined by $\xi_x =e^{-b \eta_x}$ is distributed according to the probability measure $\nu_{\beta, \lambda}$ on $(0,+\infty)^{\bb Z}$ given by
\begin{equation*}
\nu_{\beta, \lambda} (d\xi) =\prod_{x\in \bb Z} {Z}^{-1} (\beta,\lambda) {\bf 1}_{\{\xi_{x} >0\}} \exp \{ -\beta \xi_x + \lambda \log (\xi_x)\}d\xi_x
\end{equation*}
with $Z(\beta,\lambda)$ the partition function,
\begin{equation}
\label{eq:relblbl}
\beta={\bar \beta}, \quad\textrm{and}\quad \lambda= -1+{\bar \lambda}/{b}.
\end{equation}

Remark that $\nu_{\beta,\lambda}$ is nothing but a product probability measure whose marginal follows a Gamma distribution $\gamma_{\lambda+1, \beta^{-1}}$ with parameter $(\lambda+1, \beta^{-1})$. In particular, we have $Z:=Z(\beta, \lambda)= \beta^{-(\lambda +1)} \, \Gamma (\lambda +1)$, where $\Gamma$ is the usual Gamma function.

Thus, the process $(\xi (t))_{t \ge 0}$ has a family of translation invariant measures $\nu_{\beta, \lambda}$ parameterized by the chemical potentials $(\beta, \lambda)\in (0,+\infty) \times (-1, +\infty)$.

Let $\mathbb{P}_{\nu_{\beta,\lambda}}$ be the probability measure on ${D}(\mathbb{R}^+,(0,+\infty)^{\mathbb{Z}})$ induced by $(\xi(t))_{t\geq{0}}$ and $\nu_{\beta,\lambda}$ and let $\mathbb{E}_{\nu_{ \beta,\lambda}}$ denote the expectation with respect to $\mathbb{P}_{\nu_{ \beta, \lambda}}$.

Let $\langle \cdot \rangle_{\nu_{\beta, \lambda}}$ denote the average with respect to ${\nu_{\beta, \lambda}}$.
The averages $\rho:=\rho(\beta, \lambda)$ and $\theta:=\theta (\beta, \lambda)$ of the conserved quantities for $(\xi (t))_{t \ge 0}$ at equilibrium under $\nu_{\beta,\lambda}$ are defined by $\rho = \langle \xi_x \rangle_{\nu_{\beta, \lambda}}$ and $\theta = \langle \log(\xi_x) \rangle_{\nu_{\beta, \lambda}}.$ By a direct computation we get

\begin{equation}
\label{eq:chimpot}
\rho=1+{\bar e} -b {\bar v} =\cfrac{\lambda +1}{\beta}, \quad\textrm{and}\quad \theta=-b {\bar v}=\frac{\Gamma'  (\lambda +1)}{\Gamma' (\lambda+1)}- \log(\beta).
\end{equation}

It is understood, here and in the whole paper, that $(\beta, \lambda)$ are related to $(\bar \beta, \bar \lambda)$ through (\ref{eq:relblbl}). We will use the following compact notation, for each $x\in \bb Z$,
\[
{\omega}_x= \left(
\begin{array}{c}
\xi_x \\
\log (\xi_x)
\end{array}
\right),\quad \textrm{and} \quad
{\omega}= \left(
\begin{array}{c}
\rho \\
\theta
\end{array}
\right).
\]

Observe that ${\bar \omega}_x = \Lambda \omega_x - \left( \begin{array}{c} 1\\ 0 \end{array} \right)$, where
\begin{equation} \label{lambda matrix}
\Lambda=
\left(
\begin{array}{cc}
1&-1 \\
0 & -1/b
\end{array}
\right).
\end{equation}

The covariance matrix $\chi:=\chi (\beta, \lambda)$ of $\omega_0$ under $\nu_{\beta,\lambda}$ is given by
\begin{equation*}
\chi =
 \left(
\begin{array}{cc}
 \langle(\xi_0-\rho)^2 \rangle_{\nu_{\beta,\lambda}} &  \langle(\xi_0-\rho)(\log(\xi_0)-\theta) \rangle_{\nu_{\beta,\lambda}}\\
 \langle(\xi_0-\rho)(\log(\xi_0)-\theta) \rangle_{\nu_{\beta,\lambda}}&  \langle(\log(\xi_0)-\theta)^2 \rangle_{\nu_{\beta,\lambda}}
\end{array}
\right).
\end{equation*}
A simple computation shows that
\begin{equation*}
\chi =
 \left(
\begin{array}{cc}
\frac{\lambda+1}{ \beta^{2}} & \frac{1}{\beta}\\
\frac{1}{\beta} & (\log \Gamma)'' (\lambda+1)
\end{array}
\right)=
\left(
\begin{array}{cc}
\partial_\beta^2 \log (Z) & - \partial_{\beta, \lambda} \log (Z)\\
 - \partial_{\beta, \lambda} \log  (Z) & \partial^2_{\lambda} \log (Z)
\end{array}
\right).
\end{equation*}
 Denote the covariance matrix of $\bar \omega_0$ under $\mu_{\bar\beta,\bar\lambda}$ by $\bar\chi:=\bar\chi(\bar\beta,\bar\lambda)$, which is defined by
\begin{equation*}
\bar\chi =
 \left(
\begin{array}{cc}
 \langle(V_b(\eta_0)-\bar e)^2 \rangle_{\mu_{\bar\beta,\bar\lambda}} &  \langle(V_b(\eta_0)-\bar e)(\eta_0-\bar v) \rangle_{\nu_{\bar\beta,\bar\lambda}}\\
 \langle(V_b(\eta_0)-\bar e)(\eta_0-\bar v) \rangle_{\mu_{\bar\beta,\bar\lambda}}&  \langle(\eta_0-\bar v)^2 \rangle_{\mu_{\bar\beta,\bar\lambda}}
\end{array}
\right).
\end{equation*}

Thus, the covariance matrix $\chi$ of $\omega_0$ under $\nu_{\beta,\lambda}$ is related to the covariance matrix ${\bar \chi}$ of $\bar \omega_0$ under $\mu_{\bar \beta, \bar \lambda}$, by
\begin{equation}
\label{eq:chibar2}
{\bar \chi} = \Lambda \chi \Lambda^T=
\left(
\begin{array}{cc}
\frac{\lambda+1}{ \beta^{2}} +\frac{2}{\beta}+(\log \Gamma)'' (\lambda+1) & \quad \frac{1}{b\beta}+\frac{(\log \Gamma)'' (\lambda+1)}{b}\\ 
\frac{1}{b\beta}+\frac{(\log \Gamma)'' (\lambda+1)}{b}& \frac{(\log \Gamma)'' (\lambda+1)}{b^2}
\end{array}
\right).
\end{equation}

A simple computation shows that $\langle j_{x,x+1} \rangle_{\nu_{\beta, \lambda}}=-b^2\rho^2$ and $\langle j'_{x,x+1} \rangle_{\nu_{\beta, \lambda}}=-2b^2\rho$.
The hydrodynamical equations for the process $(\xi (t))_{t \ge 0}$ are given by
\begin{equation}
\label{eq:hl-ss001}
\begin{cases}
\partial_t \rho - b^2 \partial_q (\rho^2) =0\\
\partial_t \theta -2 b^2 \partial_q \rho =0
\end{cases}
\end{equation}
and can be written in the compact form $\partial_t {\mf X} + \partial_q {\mf J} ({\mf X}) =0$ with
\begin{equation*}
{\mf X}=
\left(
\begin{array}{c}
{\rho}\\
{\theta}
\end{array}
\right)
, \quad\textrm{and}\quad
{\mf J} ({\mf X})=
\left(
\begin{array}{c}
-b^2 \rho^2\\
-2b^2\rho
\end{array}
\right).
\end{equation*}
  The differential matrix of ${\mf J}$ is given by
\begin{equation*}
\nabla {\mf J}({\mf X})=
 \left(
\begin{array}{cc}
-2b^2\rho& 0\\
-2b^2 & 0
\end{array}
\right).
\end{equation*}
As above, let $({T}^+_t )_{t \ge 0}$ (resp. $({T}^{-}_t)_{t \ge 0}$) denote the semigroup on $S(\RR) \times S(\RR)$ generated by
\begin{equation}\label{eq:lin ss01}
\partial_t \varepsilon + M^T \, \partial_q \varepsilon =0, \quad ({\text{resp.}} \; \partial_t \varepsilon - M^T \, \partial_q \varepsilon =0 ).
\end{equation}
where
\begin{equation*}
M:=M(\rho,\theta)= (\nabla {\mf J})(\omega),
\end{equation*}
$\rho$ and $\theta$ are given by \eqref{eq:chimpot}. We omit the dependence of these semigroups on $(\rho, \theta)$ for lightness of the notations.

%

We remark that the transposed linearized system of (\ref{eq:hl-ss001}) around the constant profiles $(\rho, \theta)$ is given by the first equation on the left hand side of \eqref{eq:lin ss01}. It is easy to show that ${\bar M}= \Lambda M \Lambda^{-1}$ and $\Lambda^T {\bar T}_t^- = T_t^- \Lambda^T$.

\section{Orthogonal decomposition}
\label{sec:dua}

Observe that $\nu_{\beta, \lambda}$ is a product of Gamma distributions. Let us recall that the Gamma distribution $\gamma_{\alpha, k}$ with parameter $(\alpha, k)$ is the probability distribution on $(0,+\infty)$ absolutely continuous with respect to the Lebesgue measure with density $f_{\alpha, k}$ given by
\begin{equation}
f_{\alpha,k} (q) = \Big(k^\alpha \Gamma (\alpha)\Big)^{-1} q^{\alpha-1} e^{-q/k}, \quad q>0.
\end{equation}
Thus, we have $ \nu_{\beta, \lambda} (d\xi) = \prod_{x \in \ZZ} \Big( f_{\lambda+1, \beta^{-1}} (\xi_x) d\xi_x\Big) = \prod_{x \in \ZZ} \Big(\beta f_{\lambda+1, 1} (\beta \xi_x) d\xi_x \Big)$. The generalized Laguerre polynomials $(H_{n}^{(\lambda)})_{n \ge 0}$ form an orthogonal basis of the space ${\mathbb L}^{2} ( \gamma_{\lambda+1,1})$. They satisfy the following equations:
\begin{equation}\label{relations for Laguerre functions}
\begin{split}
&H_{0}^{(\lambda)}=1,\\
&q\cfrac{d}{dq} H_{n}^{(\lambda)} = nH_{n}^{(\lambda)} -(n+\lambda) H_{n-1}^{(\lambda)},\\
& \Big(q \cfrac{d^2}{dq^2} + (\lambda+1-q) \cfrac{d}{dq} +n \Big) H_{n}^{(\lambda)} =0, \\
&(n+1)H_{n+1}^{(\lambda)}(q) = (2n+1+ \lambda -q) H_{n}^{(\lambda)} (q) -(n +\lambda) {H}_{n-1}^{(\lambda)} (q)
\end{split}
\end{equation}
and the normalization condition
\begin{equation*}
\int_{0}^{\infty} \Big(H_{n}^{(\lambda)} (q)\Big)^2 f_{\lambda+1,1} (q) \, dq = \cfrac{\Gamma (\lambda +n +1)}{\Gamma (\lambda +1)}\cfrac{1}{n!}.
\end{equation*}
In particular, we have
\begin{equation}
\label{eq:fop}
\begin{split}
&H_1^{(\lambda)} (q) = -q + (\lambda +1),\\
&H_{2}^{(\lambda)} (q) = \cfrac{(2+\lambda)(1+ \lambda)}{2} - (\lambda +2) q +\cfrac{q^2}{2}.
\end{split}
\end{equation}

Let $\Sigma$ be the set composed of configurations  $\sigma=(\sigma_x)_{x\in \ZZ} \in \NN^{\ZZ}$ such that $\sigma_{x} \ne 0$ only for a finite number of $x$. The number $\sum_{x \in \ZZ} \sigma_x$ is called the size of $\sigma$ and  is denoted by $|\sigma|$. Let $\Sigma_n = \{ \sigma \in \Sigma \; ; \; |\sigma|=n\}$. On the set of $n$-tuples $\bx:=(x_1, \ldots,x_n)$ of $\ZZ^n$, we introduce the equivalence relation $\bx \sim \by$ if there exists a permutation $p$ on $\{1, \ldots,n\}$ such that $x_{p(i)} =y_i$ for all $i \in \{1, \ldots,n\}$. The class of $\bx$ for the relation $\sim$ is denoted by $[\bx]$ and its cardinal by $c({\bf x})$.  Then the set of configurations of $\Sigma_n$ can be identified with the set of $n$-tuples classes for $\sim$ by the one-to-one application:
\begin{equation*}
[{\bf x}]=[(x_1,\ldots,x_n)] \in \ZZ^n/ \sim \; \rightarrow \sigma^{[{\bf x}]} \in \Sigma_n
\end{equation*}
where for any $y \in \ZZ$, $(\sigma^{[\bf x]})_y= \sum_{i=1}^n {\bf 1}_{y=x_i}$. We will identify $\sigma \in \Sigma_n$ with the occupation numbers of  a configuration with $n$ particles, and $[\bf x]$ will correspond to  the positions of those $n$ particles.

To any $\sigma \in \Sigma$, we associate the polynomial function $H^{\beta,\lambda}_{\sigma}$ given by
\begin{equation*}
H^{\beta, \lambda}_{\sigma} (\xi) = \prod_{x \in \ZZ} H_{\sigma_x}^{(\lambda)} (\beta \xi_x).
\end{equation*}
Then, the family $\left\{ H^{\beta,\lambda}_{\sigma} \; ; \; \sigma \in \Sigma \right\}$ forms an orthogonal basis of ${\bb L}^2 (\nu_{\beta, \lambda})$ such that
\begin{equation}\label{eq:prod Hsigma}
\int H^{\beta,\lambda}_\sigma \, H^{\beta,\lambda}_{\sigma'} \, d\nu_{\beta, \lambda} = \delta_{\sigma=\sigma'} {\prod}_{x \in \ZZ}  \cfrac{\Gamma (\lambda +\sigma_x +1)}{\Gamma (\lambda+1)}\cfrac{1}{\sigma_x !} = \delta_{\sigma=\sigma'} {\mc W}^{ \lambda} (\sigma),
\end{equation}
where
\begin{equation}\label{def W}
{\mc W}^{\lambda} (\sigma):={\prod}_{x \in \ZZ}  \cfrac{\Gamma (\lambda +\sigma_x +1)}{\Gamma (\lambda+1)}\cfrac{1}{\sigma_x !}
\end{equation}
and $\delta$ denotes the Kronecker function, so that $\delta_{\sigma=\sigma'}=1$ if $\sigma=\sigma'$, otherwise it is equal to zero.

A function $F:\Sigma \to \RR$ such that $F(\sigma)=0$ if $\sigma \notin \Sigma_n$ is called a degree $n$ function. Thus, such a function is sometimes considered as a function defined only on $\Sigma_n$. A local function $f \in {\mathbb L}^2 (\nu_{\beta,\lambda})$ whose decomposition on the orthogonal basis $\{ H_{\sigma}^{\beta,\lambda} \, ; \, \sigma \in \Sigma \}$ is given by $f=\sum_{\sigma} F(\sigma)  H_{\sigma}^{\beta,\lambda}$ is called of degree $n$ if and only if $F$ is of degree $n$. A function $F: \Sigma_n \to \RR$ is nothing but a symmetric function $F:\ZZ^n \to \RR$ through the identification of $\sigma$ with $[\bx]$. We denote by $\langle \cdot, \cdot \rangle$ the scalar product on $\oplus {\mathbb L}^2 (\Sigma_n)$, each $\Sigma_n$ being equipped with the counting measure. Hence, if $F,G:\Sigma \to \RR$, we have
\begin{equation*}
\langle F, G \rangle = \sum_{n\ge 0} \sum_{\sigma \in \Sigma_n} F_n (\sigma) G_n (\sigma) = \sum_{n \ge 0} \sum_{\bx \in \ZZ^n} \frac{1}{c({\bf x})} \,  F_n (\bx) G_n (\bx),
\end{equation*}
with $F_n, G_n$ the restrictions of $F,G$ to $\Sigma_n$. We recall that $c({\bf x})$ is the cardinal of $[{\bf x}]$. Since $(\beta,\lambda)$ are fixed through the paper we denote $H_{\sigma}^{\beta,\lambda}$ by $H_{\sigma}$ and ${\mc W}^\lambda (\sigma)$ by ${\mc W} (\sigma)$.

If a local function $f \in {\bb L}^{2} (\nu_{\beta, \lambda})$ is written in the form $f(\xi)=\sum_{\sigma \in \Sigma} F(\sigma) H_{\sigma}(\xi)$ then we have
\begin{equation*}
({\mc A} f)(\xi)  = \sum_{\sigma\in\Sigma } ({\mf A} F)(\sigma) H_{\sigma}(\xi), \quad  ({\mc S} f)(\xi)  = \sum_{\sigma\in\Sigma} ({\mf S} F)(\sigma) H_{\sigma}(\xi)
\end{equation*}
with
\begin{equation*}
({\mf S} F)(\sigma) = \sum_{x \in \ZZ} ( F(\sigma^{x,x+1}) - F(\sigma)),
\end{equation*}
where $\sigma^{x,x+1}$ is obtained from $\sigma$ by exchanging the occupation numbers $\sigma_x$ and $\sigma_{x+1}$.

Let us now compute the operator ${\mf A}$. We have
\begin{equation*}
({\mc A} H_{\sigma})(\xi) = \sum_{x\in \ZZ} \xi_x (\xi_{x+1} -\xi_{x-1})\partial_{\xi_x} H_{\sigma} (\xi).
\end{equation*}
By the definition of $H_\sigma$ and by the second equality in \eqref{relations for Laguerre functions}, it follows that
\begin{equation*}
({\mc A} H_{\sigma})(\xi) = \beta \sum_{x\in\ZZ} (\xi_{x+1} -\xi_{x-1}) \Big(\sigma_{x} H_{\sigma} (\xi)- (\sigma_x +\lambda) H_{\sigma - \delta_x} (\xi)\Big),
\end{equation*}
where $\sigma-\delta_x$ is the configuration where a particle has been deleted at site $x$ (if there was no particle on site $x$, then $\sigma -\delta_x = \sigma$).

Now, noticing that the fourth equality in \eqref{relations for Laguerre functions} can be written as
$$ \beta q H_{n}^{(\lambda)} (\beta q) = (2n+1 + \lambda) H_{n}^{(\lambda)} (\beta q) - (n+\lambda) H_{n-1}^{(\lambda)} (\beta q) - (n+1) H_{n+1}^{(\lambda)} (\beta q)$$
and performing some change of variables, we have that
\begin{equation*}
\begin{split}
({\mc A} H_{\sigma})(\xi)& = \sum_{\substack{x,y \in \ZZ\\ |x-y|=1}} a(y-x) (\sigma_x + \lambda) (\sigma_y +1) H_{\sigma + \delta_{y} -\delta_x}(\xi)\\
& -\sum_{x\in\ZZ} (\sigma_x +\lambda) (\sigma_{x+1} -\sigma_{x-1}) \, H_{\sigma -\delta_{x}}(\xi) \\
& + \sum_{x\in\ZZ} (\sigma_x +1) (\sigma_{x+1} -\sigma_{x-1}) \, H_{\sigma + \delta_{x}}(\xi).
\end{split}
\end{equation*}

Here, $a(z)=-1$ if $z=-1$, $a(z)=1$ if $z=1$ and $0$ otherwise. It follows that
$${\mf A} ={\mf A}_0 + {\mf A}_- +{\mf A}_+$$
with
\begin{equation*}
\begin{split}
&({\mf A}_0 F) (\sigma) = -\sum_{\substack{x,y \in \ZZ\\ |x-y|=1}} a(y-x) \sigma_x (\sigma_y +1 + \lambda) F(\sigma+\delta_y -\delta_x),\\
&({\mf A}_+ F)(\sigma)= -\sum_{x\in\ZZ} \sigma_x (\sigma_{x+1} -\sigma_{x-1}) F(\sigma-\delta_x),\\
&({\mf A}_- F)(\sigma)= \sum_{x\in\ZZ} (\sigma_x-1+\lambda) (\sigma_{x+1} -\sigma_{x-1}) F(\sigma+\delta_x).
\end{split}
\end{equation*}
Observe that if $F$ vanishes outside of $\Sigma_n$ then ${\mf A}_{\pm } F$ vanishes outside of $\Sigma_{n \mp 1}$ and ${\mf A}_0$ vanishes outside of $\Sigma_n$. In other words, ${\mf A}_0 $ keeps fixed the degree of a function, ${\mf A}_+ $ raises the degree by one while ${\mf A}_-$ lowers the degree by one.


The Dirichlet form ${\mc D} (f)$ of a local function $f \in {\mathbb L}^2 (\nu_{\beta,\lambda})$ is defined by
\begin{equation*}
{\mc D} (f)=\langle f \,, (-{\mc S} f) \rangle_{\nu_{\beta,\lambda}}= \cfrac{1}{2} \sum_{x \in \ZZ}  \int \left( f(\xi^{x,x+1}) - f(\xi)  \right)^2\nu_{\beta,\lambda}(d \xi).
\end{equation*}
Recall that $\langle \cdot,\cdot\rangle_{\nu_{\beta,\lambda}}$ denotes the inner product of $\bb L^2(\nu_{\beta,\lambda})$.

Since $f$ has the decomposition $f= \sum_{\sigma \in \Sigma} F(\sigma) H_{\sigma}$ then
\begin{equation}
\label{eq:df01}
{\mc D} (f) =\cfrac{1}{2} \sum_{x \in \ZZ} \sum_{\sigma \in \Sigma} {\mc W} (\sigma) \left( F(\sigma^{x,x+1}) -F(\sigma) \right)^2.
\end{equation}
Let $\Delta_{+}= \left\{ (x,y) \in \ZZ^2 \, ; \, y \ge x+1\right\}$, $\Delta_- = \left\{ (x,y) \in \ZZ^2 \, ; \, y \le x-1\right\}$ and $\Delta_0 = \left\{ (x,x) \, ; \, x \in \ZZ \right\}$. We denote by ${\bb D}_1$ the Dirichlet form of a symmetric simple one dimensional random walk, i.e.
\begin{equation*}
{\bb D}_1 (F) = \cfrac{1}{2} \sum_{x \in \ZZ} (F(x+1) - F(x))^2,
\end{equation*}
where $F: \ZZ \to \RR$ is such that $\sum_{x\in \bb Z} F^2 (x) < \infty$.

We denote by ${\bb D}_2$ the Dirichlet form of a symmetric simple random walk on $\ZZ^2$ where jumps from $\Delta_{\pm}$ to $\Delta_0$ and from $\Delta_0$ to $\Delta_{\pm}$ have been suppressed and jumps from $(x,x) \in \Delta_0$ to $(x \pm 1, x \pm 1) \in \Delta_0$ have been added, i.e.
 \begin{equation*}
{\bb D}_2(F)=\cfrac{1}{2} \sum_{|\be|=1} \!\sum_{\bx \in \Delta_{\pm}, \bx + \be \in \Delta_{\pm}}\!\! \!\!\!\!\left( F(\bx + \be) -F(\bx)\right)^2 +  \cfrac{1}{2} \sum_{\bx \in \Delta_0} \!\left( F(\bx \pm (1,1)) -F(\bx)\right)^2,
\end{equation*}
where $F:\ZZ^2 \to \RR$ is a symmetric function such that $\sum_{\bx \in \ZZ^2} F^2 (\bx) < \infty$.

\begin{lemma}
\label{lem:df12}
Let $f=\sum_{n=1}^2 \sum_{\sigma \in \Sigma_n} F_n (\sigma) H_{\sigma}$ be a local  function  such that $F_1$ (resp. $F_2$) is of degree $1$ (resp. degree $2$). There exists a positive constant $C:=C(\lambda)$, independent of $f$, such that
\begin{equation*}
C^{-1} \left[ {\bb D}_1 (F_1) +{\bb D}_2 (F_2) \right]  \le {\mc D} (f) \le C \left[ {\bb D}_1 (F_1) +{\bb D}_2 (F_2)  \right].
\end{equation*}
\end{lemma}

\begin{proof}
Observe that
\begin{itemize}
\item If $\sigma \in \Sigma_1$, then ${\mc W} (\sigma)=(\lambda+1)$.
\item If $\sigma \in \Sigma_2$, $\sigma= \delta_x +\delta_y$, $x\ne y$, then ${\mc W} (\sigma)=(\lambda+1)^2$; if $\sigma \in \Sigma_2$, $\sigma=2 \delta_x$, then ${\mc W} (\sigma)= [(\lambda +2)(\lambda +1)]/2$.
\end{itemize}
This follows from the relation $\Gamma (z+1)=z\Gamma (z)$. Then, by using (\ref{eq:df01}) and the identification of functions $F: \Sigma_n \to \RR$ of degree $n$ with their representations as symmetric functions on $\ZZ^n$, the claim follows.
\end{proof}

\section{Triviality of the fluctuations}
\label{sec:triv}
In this section we prove Theorems \ref{th:fluct-hs} and \ref{th:fluct-ds} and Corollaries \ref{CLT Energy flux}, \ref{CLT volume flux} and \ref{vanishing of Energy flux} above.  The proof of Theorems \ref{th:fluct-hs} and \ref{th:fluct-ds} is standard and relies on a careful analysis of martingales associated to the respective density fields. For this reason we present  only the sketch of their proofs. For the interested reader we refer to chapter 11 of \cite{KL}.  We notice that the restrictions on $k$ appearing  in the statement of those theorems come from tightness estimates, that we do not prove here since they follow from very similar computations to those presented in \cite{KL}.

To approach the proof of theorems we notice that since $V_b (\eta_x) -1= \xi_x -\log (\xi_x)$, $\eta_x = -b^{-1} \log (\xi_x)$, the problem is reduced to study the fluctuation field of the conserved quantities for the process $(\xi (t))_{t \ge 0}$ at equilibrium under the probability measure $\nu_{\beta, \lambda}$. The fluctuation field  for $(\xi (t))_{t \ge 0}$ is defined by
\begin{equation}
  \label{eq:ZZ}
  \mathcal{Z}^{n,\alpha}_t  (\bG) =\frac{1}{\sqrt{n}} \sum_{x\in \ZZ}
  \bG\left(x/n\right) \cdot \left({\omega}_x (tn^{1+\alpha})  - {\omega}\right),
\end{equation}
where $\bG$ is a  test function belonging to $S(\RR) \times S(\RR)$. Recalling (\ref{lambda matrix}) we have
\begin{equation}\label{rel Y and Z}
{\mc Y}_t^{n,\alpha} (\bG) = \cfrac{1}{\sqrt n} \sum_{x \in \ZZ} (\Lambda^T\bG) (x/n) \cdot (\omega_x (t n^{1+\alpha})  -\omega)={\mc Z}_t^{n, \alpha} (\Lambda^T \bG).
\end{equation}

By the relation ${\bar M}= \Lambda M \Lambda^{-1}$, we are able to translate any result about the convergence of ${\mc Z}_\cdot^{n,\alpha}$ into a corresponding result for ${\mc Y}_{\cdot}^{n,\alpha}$.

\subsection{The hyperbolic scaling}

For any local function $g:= g(\xi)$ we define the projection ${\mc P}_{\rho, \theta} \, g$ of $g$ on the fields of the conserved quantities by
\begin{equation*}
({\mc P}_{\rho, \theta} g)(\xi)=(\nabla {\tilde g})(\rho, \theta) \cdot (\omega_0-\omega)
\end{equation*}
where ${\tilde g} (\rho,\theta) = \langle g \rangle_{\nu_{\beta, \lambda}}$ and  $\nabla {\tilde g}$ is the gradient of the function ${\tilde g}$.

We have that
\begin{prop}[Boltzmann-Gibbs principle I]

For every $\mb H \in S(\RR)\times S(\RR)$ and every $t>0$,
\begin{equation*}
\lim_{n \to \infty} {\mathbb E}_{\nu_{\beta,\lambda}} \left[ \left( \int_0^{t} \cfrac{1}{\sqrt{n}} \sum_{x \in \ZZ} \mb H\left(x/n\right)\cdot \left[\tau_x V_{J_{0,1}}(\xi(sn))\right] ds\right)^2\right] =0,
\end{equation*}
where for a local function $g$  we define $V_g(\xi):= g(\xi)-\tilde{g}(\rho,\theta)- {\mc P}_{\rho,  \theta}  g (\xi )$
and for $\xi\in{(0,+\infty)^\bb Z}$, $\tau_xg(\xi):=g(\tau_x\xi)$, $\tau_x\xi(y):=\xi(x+y)$ and $J_{0,1}$ is given in \eqref{flux for xi}.
\end{prop}

Since we prove a refined version of this proposition we omit its proof. As a consequence of last result, we get that the fluctuation field $({\mc Z}_{\cdot}^{n,0})_{n \ge 1}$ converges in law (in the sense of Theorem  \ref{th:fluct-hs}) to $\mathcal{Z}_\cdot^0$  solution of the equation at the right hand side of (\ref{eq:lin ss01}). Theorem \ref{th:fluct-hs} is a simple consequence of this fact.

In order to prove Corollary \ref{CLT Energy flux} and Corollary \ref{CLT volume flux} we follow the approach first presented in \cite{RV} and considered also in \cite{J.L.} (resp. \cite{G}) for the symmetric (resp. asymmetric) simple exclusion. For that reason we sketch the main steps of the proof.  For more details we refer the reader to, for example, the proof of Theorem 4.2 of \cite{G}. The main goal is to related the energy and volume currents with the density field and to use the result of Theorem \ref{th:fluct-hs}.  For that purpose and whenever the total energy (resp. volume) at $\eta$ is finite we can write down the energy (resp. volume) flux through the bond $\{x,x+1\}$ during the time interval $[0,tn]$, as:
\begin{equation*}
\mathcal{E}_{x,x+1}^n(t):=\sum_{y\geq{x+1}}\Big\{V_b(\eta_y(tn))-V_b(\eta_y(0))\Big\}
\end{equation*}
\begin{equation*}
 \Big( \text{resp. \; } \mathcal{V}_{x,x+1}^n(t):=\sum_{y\geq{x+1}}\Big\{\eta_y(tn)-\eta_y(0)\Big\}\Big).
\end{equation*}

 In such case, we can relate the energy (resp. volume) flux given above with the energy-volume fluctuation field as
\begin{equation*}
\mathcal{E}_{x,x+1}^n(t):=\mathcal{Y}_t^n(H_x^1)-\mathcal{Y}^n_0(H_x^1)
\end{equation*}
\begin{equation*}
 \Big( \textrm{resp. \; } \mathcal{V}_{x,x+1}^n(t):=\mathcal{Y}_t^n(H_x^2)-\mathcal{Y}^n_0(H_x^2),
\end{equation*}
where
\[
H_x^1(y) = \left(
\begin{array}{c}
\bf 1_{\{y\geq{x}\}} \\
0
\end{array}
\right), \quad
H_x^2(y) = \left(
\begin{array}{c}
0 \\
\bf 1_{\{y\geq{x}\}}
\end{array}
\right).
\]
Since the function $\bf 1_{\{y\geq{x}\}}$ does not belong to our space of test functions for which we derived Theorem \ref{th:fluct-hs} we first show that
\begin{proposition} \label{prop nec}
For every $t\geq{0}$,
\begin{equation*}
\lim_{\ell\to\infty} \mathbb{E}_{\nu_{\beta,\lambda}}\Big[\Big(\mathcal{E}_{x,x+1}^n(t)-(\mathcal{Y}_t^n(G_{\ell,x}^1)-\mathcal{Y}^n_0(G_{\ell,x}^1))\Big)^2\Big]=0,
\end{equation*}
\begin{equation*}
(\textrm{resp.} \lim_{\ell\to\infty} \mathbb{E}_{\nu_{\beta,\lambda}}\Big[\Big(\mathcal{V}_{x,x+1}^n(t)-(\mathcal{Y}_t^n(G_{\ell,x}^2)-\mathcal{Y}^n_0(G_{\ell,x}^2))\Big)^2\Big]=0,
\end{equation*}
where \[
G_{\ell,x}^1(y) = \left(
\begin{array}{c}
G_{\ell,x}(y) \\
0
\end{array}
\right), \quad
G_{\ell,x}^2(y) = \left(
\begin{array}{c}
0 \\
G_{\ell,x}(y)
\end{array}
\right)
\]
and $G_{\ell,x}(y):=(1-y/\ell)\bf 1_{\{x\leq{y}\leq{x+\ell}\}}$.
\end{proposition}

The proof of last result follows the same lines as in the proof of Proposition 4.1 of \cite{G} and for that reason we omitted it.
We notice that, at this point we are still not able to apply Theorem \ref{th:fluct-hs} since $G_{\ell,x}^1$ and $G_{\ell,x}^2$ are not functions in $\mathcal{S}(\bb R)$. Therefore, we approximate in $L^2(\bb R)$ each one of these functions by smooth functions for which Theorem \ref{th:fluct-hs} holds. Then, the proof of Corollary \ref{CLT Energy flux} and \ref{CLT volume flux} follows combining the previous proposition with Theorem \ref{th:fluct-hs}. For more details on this argument, we refer the reader to \cite{G}.

Finally, in order to compute the limiting variance, for example for the energy flux, we do the following. Here we take $x=0$ to simplify the notation
\begin{equation*}
\begin{split}
{\bb E_{Q}}[Z^e_{t}Z^e_{s}]&={\bb E_{Q}}[\{\mathcal{Y}_{t}(H_0^1)-\mathcal{Y}_{0}(H_0^1)\}\{\mathcal{Y}_{s}(H_0^1)-\mathcal{Y}_{0}(H_0^1)\}]\\
&\hspace{0.1cm}\begin{split}=\lim_{\ell\rightarrow{\infty}}{\bb E_{Q}}\Big[\mathcal{Y}_{t}(G_{\ell,0}^1)\mathcal{Y}_{s}(G_{\ell,0}^1)-&\mathcal{Y}_{t}(G_{\ell,0}^1)\mathcal{Y}_{0}(G_{\ell,0}^1)\\
-\mathcal{Y}_{s}(G_{\ell,0}^1)&\mathcal{Y}_{0}(G_{\ell,0}^1)+\mathcal{Y}_{0}(G_{\ell,0}^1)\mathcal{Y}_{0}(G_{\ell,0}^1)\Big]\end{split}
\end{split}
\end{equation*}

Now, to compute last expectation we use the change of variables. Notice that for $H,G\in{\mathcal{S}(\bb R)}$ we have that ${\bb E_{Q}}[\mathcal{Z}^0_t(H)\mathcal{Z}^0_0(G)]:=\left<T^{-}_t H \cdot \chi G\right>$. Combining this with \eqref{rel Y and Z},  it follows that ${\bb E_{Q}}[\mathcal{Y}_t(H)\mathcal{Y}_0(G)]:=\left<T^{-}_t(\Lambda^T H)\cdot \chi\Lambda^TG\right>$.
By the definition of $(T_t^{-})_{t\geq{0}}$ we have for $G_1, G_2$ test functions in $S(\RR)$:
\[
T^{-}_{t}\left(
\begin{array}{c}
G_1(x) \\
G_2(x)
\end{array}
\right)
 = \left(
\begin{array}{c}
\frac{1}{\rho}\Big(G_2(x-2b^2\rho t)-G_2(x)\Big)+G_1(x-2b^2\rho t) \\
G_2(x)
\end{array}
\right). \quad
\]
As a consequence we obtain that
\begin{equation*}
\begin{split}
{\bb E_{Q}}[&Z^e_{t}Z^e_{s}]=\Big(1-{\frac{1}{\rho}}\Big)^2\Big(\frac{\lambda+1}{\beta^2}\Big)\\
&\times\lim_{\ell\rightarrow{\infty}}\int_{\mathbb{R}}\Big(G_\ell^t(x)G_\ell^s(x)-G_\ell^t(x)G_\ell(x)-G_\ell^s(x)G_\ell(x)+G_\ell(x)G_\ell(x)\Big)dx,
\end{split}
\end{equation*}
where for $t\geq{0}$, $G_\ell^{\,t}(x):=G_{\ell,0}(x-2b^2\rho t)$.
Now, using \eqref{eq:relblbl} and \eqref{eq:chimpot} the proof ends. Analogously, repeating the computations above, replacing $H_0^1$ by $H_0^2$ we get the covariance for the volume flux.

\subsection{The longer time scale}

Since in the hyperbolic time scale the initial fluctuations for the field $\mathcal{Z}^{n,\alpha}_\cdot$ are transported by the transposed  linearized system given on the right hand side of (\ref{eq:lin ss01}), we redefine the fluctuation field ${\widehat {\mc Z}}_{\cdot}^{n,\alpha}$, $\alpha>0$, on $\bG \in S(\RR) \times S(\RR)$, by
\begin{equation*}
{\widehat {\mc Z}}_t^{n,\alpha} (\bG)= {\mc Z}_t^{n, \alpha} \left( {T}^{+}_{tn^{\alpha}} \bG \right).
\end{equation*}

%
By Dynkin's formula, see for example Appendix 1, section 5 of \cite{KL}
\begin{equation*}
{\mc M}_t^{n, \alpha} (\bG) = {\widehat {\mc Z}}_t^{n,\alpha} (\bG) - {\widehat {\mc Z}}_0^{n,\alpha} (\bG) - \int_{0}^t \left\{ n^{1+\alpha} {\mc L} \left( {\widehat {\mc Z}}_s^{n,\alpha} (\bG)\right) + \partial_s {\widehat {\mc Z}}_s^{n,\alpha} (\bG)\right\} \, ds
\end{equation*}
is a martingale with quadratic variation given by
\begin{equation*}
\langle {\mc M}^{n, \alpha} \rangle_t=\int_{0}^t  n^{1+\alpha} {\mc L} \left( {\widehat {\mc Z}}_s^{n,\alpha} (\bG)\right)^2-2n^{1+\alpha} \left( {\widehat {\mc Z}}_s^{n,\alpha} (\bG)\right){\mc L} \left( {\widehat {\mc Z}}_s^{n,\alpha} (\bG)\right) \, ds.
\end{equation*}
A simple computation shows that ${\mathbb E}_{\nu_{\beta,\lambda}}[\langle {\mc M}^{n, \alpha} \rangle_t]$ vanishes as $n$ goes to ${\infty}$ for $\alpha<1$. This is equivalent to saying that the martingale ${\mc M}^{n, \alpha}_t$ vanishes as $n$ goes to ${\infty}$ in ${\mathbb L}^{2} (\PP_{\nu_{\beta,\lambda}})$, for $\alpha<1$.
Observe that, by definition of $(T^{+}_t)_{t \ge 0}$, we have
\begin{equation*}
\begin{split}
\partial_s {\widehat {\mc Z}}_s^{n,\alpha} (\bG) &= -\cfrac{n^\alpha}{\sqrt n} \sum_{x \in \ZZ} M^T \left[ \partial_q\left( T^+_{s n^{\alpha}} \bG \right) (x/n) \right]  \cdot (\omega_x (s n^{1+\alpha}) -\omega)\\
&=-\cfrac{n^\alpha}{\sqrt n} \sum_{x \in \ZZ}  \left[ \partial_q\left( T^+_{s n^{\alpha}} \bG \right) (x/n) \right]  \cdot M(\omega_x (s n^{1+\alpha}) -\omega).
\end{split}
\end{equation*}
On the other hand, the first term in the integral part of the martingale ${\mc M}_t^{n, \alpha} (\bG)$ is equal to
\begin{equation*}
\cfrac{n^\alpha}{\sqrt n} \sum_{x \in \ZZ} n\Big( (T^+_{sn^{\alpha}} \bG) \left( \cfrac{x+1}{n}\right) - (T^+_{sn^{\alpha}} \bG) \left( \cfrac{x}{n} \right)\Big) \cdot \Big( J_{x,x+1} (\xi(sn^{1+\alpha})) \,  - \, \langle J_{x,x+1} \rangle_{}\nu_{\beta,\lambda}\Big).
\end{equation*}
Performing a Taylor expansion, we can replace this term, up to a term vanishing as $n$ goes to $\infty$ in ${\mathbb L}^{2} (\PP_{\nu_{\beta,\lambda}})$,  by
\begin{equation*}
\cfrac{n^\alpha}{\sqrt n} \sum_{x \in \ZZ}  \Big((\partial_q T^+_{s n^{\alpha}} \bG) (x/n)\Big) \cdot \Big( J_{x,x+1} (\xi(sn^{1+\alpha})) \,  - \, \langle J_{x,x+1} \rangle_{}\nu_{\beta,\lambda}\Big).
\end{equation*}


Thus, in order to show that
\begin{equation}\label{vanish Z field}
\lim_{n \to \infty}{\mathbb E}_{\nu_{\beta,\lambda}} \left[ \Big({\widehat {\mc Z}}_t^{n,\alpha} (\bG) - {\widehat {\mc Z}}_0^{n,\alpha} (\bG)\Big)^2\right]=0,
 \end{equation}
 it remains to show that
\begin{equation*}
\lim_{n \to \infty} {\mathbb E}_{\nu_{\beta,\lambda}} \left[ \left(  \cfrac{n^\alpha}{\sqrt n} \int_0^t \, ds \, \sum_{x \in \ZZ} (\partial_q T^+_{s n^{\alpha}} \bG) (x/n) \cdot \Theta_{x} (\xi(sn^{1+\alpha})) \right)^2\right]=0
\end{equation*}
where for $\xi\in{(0,+\infty)^\bb Z}$
\begin{equation*}
\Theta_x(\xi)=  J_{x,x+1}(\xi) \,  - \, \langle J_{x,x+1} \rangle_{\nu_{\beta,\lambda}} -M\, (\omega_x -\omega).
\end{equation*}
Observe that in this formula, $M:=M(\rho, \theta)$ is the differential with respect to $(\rho,\theta)$ of the function $ \langle J_{x,x+1} \rangle_{\nu_{\beta,\lambda}}$ as computed below \eqref{eq:lin ss01}.  A simple computation shows that for $\xi\in{(0,+\infty)^\bb Z}$
\begin{equation*}
\Theta_x(\xi)=
\left(
\begin{array}{c}
-b^{2} (\xi_{x+1} -\rho)(\xi_x -\rho)- (\gamma + b^2 \rho) \nabla \xi_x \\
-\nabla (b^2 \xi_x + \gamma \log (\xi_x))
\end{array}
\right).
\end{equation*}

The discrete gradient terms appearing in the previous expression, permit to perform another discrete integration by parts and the resulting terms vanish in ${\mathbb L}^{2} (\PP_{\nu_{\beta,\lambda}})$ as $n$ goes to $\infty$, for $\alpha<1$. Using the smoothness of the function $\bG$, we see that it only remains to show the following theorem with $\varphi (s,q)$ equal to the first component of the column vector $\partial_q T^+_{s} \bG$.

\begin{theo}[Boltzmann-Gibbs principle II]
\label{th:BGII}

Fix $\alpha <1/3$ and let $\varphi:\mathbb{R}^+\times \RR\rightarrow{\RR}$ be such that for any $t \ge 0$, $\varphi(t,\cdot)\in S(\RR)$. For every $t>0$
\begin{equation*}
\lim_{n \to \infty} {\mathbb E}_{\nu_{\beta,\lambda}} \left[ \left(\int_0^t \cfrac{n^{\alpha}}{\sqrt{n}} \sum_{x \in \ZZ} \varphi (sn^\alpha,x/n) (\xi_x(sn^{1+\alpha}) -\rho) (\xi_{x+1}(sn^{1+\alpha}) -\rho)  \, ds \right)^2\right] =0
\end{equation*}
\end{theo}

\begin{proof}
In the following, $C, C_0, C_1, \ldots$ denote constants independent of $n$ whose values can change from line to line.

 Let $f_s(\xi)$ be the function defined by
$$f_s (\xi)= \sum_{x \in \ZZ} \varphi (s,x/n) H_{\delta_x +\delta_{x+1}}(\xi)=-\beta^2 \sum_{x \in \ZZ} \varphi (s,x/n)  (\xi_x -\rho) (\xi_{x+1} -\rho).$$
The last equality follows from \eqref{eq:fop} and \eqref{eq:chimpot}.

 We have the following upper bound
\begin{equation*}
\begin{split}
 {\mathbb E}_{\nu_{\beta,\lambda}} \left[ \left( \int_0^t\!\! f_{sn^\alpha}(\xi (s n^{1+\alpha}))\,ds \right)^2\right]  & \le C \int_0^t \langle\, f_{sn^\alpha}, (s^{-1} - n^{1+\alpha} {\mc L})^{-1} f_{sn^\alpha} \rangle_{\nu_{\beta,\lambda}}ds \\
 &= \cfrac{C} {n^{1+\alpha}} \int_{0}^t \Big\langle f_{sn^\alpha}, \left( \cfrac{1}{s n^{1+\alpha}} - {\mc L} \right)^{-1}\!\!\! \!\!f_{sn^\alpha} \Big\rangle_{\nu_{\beta,\lambda}}\!\!ds\\
 &\le   \cfrac{C} {n^{1+\alpha}}\int_{0}^t  \Big\langle  f_{sn^\alpha}, \left( \cfrac{1}{s n^{1+\alpha}} - \gamma {\mc S} \right)^{-1}\!\!\!\!\! f_{sn^\alpha} \Big \rangle_{\nu_{\beta,\lambda}}\!\!ds.
 \end{split}
\end{equation*}

In the first inequality above we used Lemma 3.9 of \cite{S.} applied to this setting. We notice that since our test functions depend on time, the lemma of \cite{S.}  has to be modified as written here. To prove the last result one can simply adapt the proof of Lemma 4.3 of \cite{C.L.O.} to this case.

In order to simplify notations, let us define $\ve= 1/sn^{1+\alpha}$.

We denote by $\Sigma_2^0$ the set of configurations $\sigma$ of $\Sigma_2$ such that $\sigma= 2 \delta_x$, $x \in \ZZ$, and $\Sigma_2^{\pm}$ the complementary set of $\Sigma_2^0$ in $\Sigma_2$, i.e. the set of configurations $\sigma \in \Sigma_2$ such that $\sigma=\delta_x + \delta_y$, $y \ne x \in \ZZ$. Observe that $f_{sn^\alpha}$ is a function of degree $2$ with a decomposition in the form $f_{sn^\alpha}=\sum_{\sigma \in \Sigma_2} \Phi _{sn^\alpha}(\sigma) H_{\sigma}$ which satisfies $\Phi_{sn^\alpha}(\sigma)=0$ if $\sigma \in \Sigma_2^0$. We have that (see e.g. \cite{S.})
\begin{equation*}
\Big\langle  f_{sn^\alpha}\, , \left( \ve-\gamma {\mc S} \right)^{-1} f_{sn^\alpha} \Big \rangle_{\nu_{\beta,\lambda}} = \sup_{g} \Big\{ 2 \langle f_{sn^\alpha}, g\rangle_{\nu_{\beta,\lambda}} - \ve \langle g\;,g\rangle_{\nu_{\beta,\lambda}} - \gamma  {\mc D} (g) \Big\}
\end{equation*}
where the supremum is taken over local functions $g \in {\mathbb L}^2 (\nu_{\beta,\lambda})$. Decompose $g$ appearing in this variational formula as $g=\sum_{\sigma} G(\sigma) H_{\sigma}$. Recall that $\{H_{\sigma} \, ; \, \sigma \in \Sigma\}$ are orthogonal, that the function $f_{sn^\alpha}$ is a degree $2$ function such that $\Phi_{sn^\alpha}(\sigma)=0$ for any $\sigma \notin \Sigma_2^{\pm}$  and formula (\ref{eq:df01}) for the Dirichlet form ${\mc D} (g)$. Thus, we can restrict this supremum over degree $2$ functions $g$ such that $G(\sigma)=0$ if $\sigma \in \Sigma_2^0$. Then, by Lemma \ref{lem:df12}, we have
\begin{equation*}
\begin{split}
\Big\langle  f_{sn^\alpha}, \left( \ve-\gamma {\mc S} \right)^{-1} f_{sn^\alpha}& \Big \rangle_{\nu_{\beta,\lambda}}\le  \sup_{G} \left\{  \sum_{x \ne y} \Phi_{sn^\alpha}(x,y) G(x,y)  - \ve \sum_{\substack{(x,y) \in \ZZ^2\\ x\neq y}} G^2 (x,y) \right.\\
&\left.\quad \quad \quad-C \sum_{|\be| =1}\sum_{\substack{(x,y) \in \Delta^{\pm}\\ (x,y) + \be \in \Delta^{\pm}}}  \Big( G((x,y)+\be) -G(x,y) \Big)^2\right\}
\end{split}
\end{equation*}
where $C:=C(\lambda, \gamma)$, $\Delta_{\pm} = \{ (x,y) \in \ZZ^2 \, ; \, x\ne y\}$ and as usual we identify the functions defined on $\Sigma_n$ with symmetric functions defined on $\ZZ^n$.

 In order to get rid of the geometric constraints appearing in the last term of the variational formula, for any symmetric function $G$ defined on the set $\Delta_{\pm}$, we denote by ${\tilde G}$ its extension to $\ZZ^2$ defined by
$${\tilde G} (x,y) =G (x,y) \; \text{if}\;  x\ne y, \quad {\tilde G} (x,x) = \cfrac{1}{4} \sum_{|\be|=1} G((x,x) +\be).$$
It is trivial that
\begin{equation*}
\begin{split}
& \sum_{(x,y) \in \ZZ^2} {\tilde G}^2 (x,y)\le C\sum_{\substack{(x,y) \in \ZZ^2\\ x\neq y}} {G}^{2} (x,y) ,\\
\textrm{and}\\
& \sum_{|\be| =1}\sum_{(x,y) \in \ZZ^2}  \Big({\tilde G} ((x,y)+\be) -{\tilde G} (x,y) \Big)^2 \\
&\quad \quad \quad \quad \quad \quad\quad \quad \quad\quad \quad \quad\le C \sum_{|\be| =1}\sum_{\substack{(x,y) \in \Delta^{\pm}\\ (x,y) + \be \in \Delta^{\pm}}}  \Big(G((x,y)+\be) -G(x,y) \Big)^2 .
\end{split}
\end{equation*}

Thus, we have
\begin{equation*}
\begin{split}
\Big\langle  f_{sn^\alpha}, \left( \ve-\gamma {\mc S} \right)^{-1}& f_{sn^\alpha}\Big \rangle_{\nu_{\beta,\lambda}}  \\
\le  C_0 \sup_{G} &\left\{  \sum_{(x,y) \in \ZZ^2} \Phi _{sn^\alpha}(x,y) G(x,y)  - C_1 \ve \sum_{(x, y) \in \ZZ^2} G^2 (x,y) \right. \\
&\left. \quad \quad \quad -C_2 \sum_{|\be| =1}\sum_{(x,y) \in \ZZ^2 }  \Big( G((x,y)+\be) -G(x,y) \Big)^2\right\}
\end{split}
\end{equation*}
where the supremum is now taken over all symmetric local functions $G:\ZZ^2 \to \RR$.  Notice that the last variational formula is equal to the resolvent norm, for a simple symmetric two dimensional random walk, of the function $\Phi_{sn^\alpha}$. By using Fourier transform one can easily show that this supremum is equal to
\begin{equation*}
\cfrac{C_0}{4} \int_{[0,1]^2} \cfrac{|{\hat \Phi}_{sn^\alpha}(\bk)|^2}{C_1 \ve + 4C_2 \sum_{i=1}^2 \sin^{2} (\pi k_i) } d\bk
\end{equation*}
where the Fourier transform ${\hat \Phi}_{sn^\alpha}$ of $\Phi_{sn^\alpha}$ is given by
\begin{equation*}
{\hat \Phi}_{sn^\alpha}( \bk) = \sum_{(x,y) \in \ZZ^2} \Phi_{sn^\alpha}(x,y) e^{2i\pi (k_1 x +k_2 y)}, \quad \bk= (k_1,k_2) \in [0,1]^2.
\end{equation*}
By definition of $f_{sn^\alpha}$, we have $\Phi_{sn^\alpha}(x,y)=\cfrac{1}{2} \Big(\varphi (sn^\alpha,x/n) + \varphi(sn^\alpha,y/n)\Big)$ if $|x-y|=1$ and $0$ otherwise. Consequently, we have
\begin{equation*}
\begin{split}
\Big\langle  f_{sn^\alpha}, &\left( \ve -\gamma {\mc S} \right)^{-1} f_{sn^\alpha} \Big \rangle_{\nu_{\beta,\lambda}}   \le \cfrac{C_0}{16} \int_{[0,1]^2} \cfrac{\left| \sum_{x\in \bb Z} \varphi (sn^\alpha,x/n) e^{2i \pi x (k_1 +k_2)} \right|^2}{C_1 \ve + 4C_2 \sum_{i=1}^2 \sin^{2} (\pi k_i) } d\bk\\
&=  \cfrac{C_0}{16} \int_{[0,2]}\!\!\Big(\int_{[0,1]}  \cfrac{\textbf{1}_{[\sup(1-p,1),\inf(1,0)]}(p)\left| \sum_{x\in \bb Z} \varphi (sn^\alpha,x/n) e^{2i \pi x p} \right|^2}{C_1 \ve + 4C_2 \sin^{2} (\pi k_1) +4C_2 \sin^{2} (\pi (p-k_1)) } dk_1 \,\Big) dp\\
&= \cfrac{C_0}{16} \int_{[0,1]}\Big( \int_{[0,1]}\cfrac{\left| \sum_{x\in \bb Z} \varphi (sn^\alpha,x/n) e^{2i \pi x p} \right|^2}{C_1 \ve + 4C_2 \sin^{2} (\pi k_1) +4C_2 \sin^{2} (\pi (p-k_1)) } dk_1 \,\Big) dp\\
\end{split}
\end{equation*}
where we used the change of variables $p=k_2+k_1$ for the first equality and the periodicity of the functions involved for the second one. It follows that
\begin{equation*}
\begin{split}
\Big\langle  f_{sn^\alpha}, \left( \ve-\gamma {\mc S} \right) &^{-1} f_{sn^\alpha} \Big \rangle_{\nu_{\beta,\lambda}} \\
& \le \cfrac{C_0}{16} \int_{[0,1]} {\left| \sum_{x\in \bb Z} \varphi (sn^\alpha,x/n) e^{2i \pi x p} \right|^2} dp \int_{[0,1]} \cfrac{dk_1}{C_1 \ve + 4C_2 \sin^{2} (\pi k_1) } \\&\le \cfrac{C}{\sqrt{\ve}} \int_{[0,1]} {\left| \sum_{x\in \bb Z} \varphi (sn^\alpha,x/n) e^{2i \pi x p} \right|^2} dp.
\end{split}
\end{equation*}
 Observe now that
\begin{equation*}
\begin{split}
\int_{[0,1]} {\left| \sum_{x\in \bb Z} \varphi (sn^\alpha x/n) e^{2i \pi x p} \right|^2}dp 
&= \sum_{x\in \bb Z} \varphi^2 (sn^{\alpha}, x/n) \le C n.
\end{split}
\end{equation*}
Putting everything together, we get that
\begin{equation*}
\begin{split}
{\mathbb E}_{\nu_{\beta,\lambda}} \Big[  \Big(\int_0^t \cfrac{n^{\alpha}}{\sqrt{n}} \sum_{x \in \ZZ} \varphi (sn^\alpha,x/n)  (\xi_x(sn^{1+\alpha}) -\rho) (\xi_{x+1}&(sn^{1+\alpha}) -\rho)\, ds\Big)^2 \Big] \\
&\le \frac{ Ct n^{2\alpha -1}}{n^{1+\alpha}} \int_{0}^{t} \cfrac{n}{\sqrt{\ve}}\;ds.
\end{split}
\end{equation*}
Since $\ve:=1/sn^{1+\alpha}$ last expression
vanishes as $n$ goes to $\infty$, if $\alpha<1/3$.
\end{proof}

Now, in order to prove Corollary \ref{vanishing of Energy flux} we follow the same arguments as in the proof of Proposition 9.3 of \cite{G} and we proceed as follows. Whenever the total energy (resp. volume) at $\eta$ is finite we can write down:
\begin{equation}\label{longer energy}
\begin{split}
& \mathcal{E}_{u_t^{x,\alpha}(n)}^n(t):=\sum_{y\geq{u_t^{x,\alpha}(n)}}\Big\{V_b(\eta_y(tn^{1+\alpha}))-V_b(\eta_y(0))\Big\}, \\
\Big(&\text{ resp. \, } \mathcal{V}_{u_t^{x,\alpha}(n)}^n(t):=\sum_{y\geq{u_t^{x,\alpha}(n)}}\Big\{\eta_y(tn^{1+\alpha})-\eta_y(0)\Big\}\Big).
\end{split}
\end{equation}
In order to justify the previous equalities one can repeat the same arguments as used in the hyperbolic scaling. Now, we use the change of variables to define the energy (resp. volume) flux through the time-dependent bond $\{u_t^{x,\alpha}(n),u_t^{x,\alpha}(n)+1\}$ during the time interval $[0,tn^{1+\alpha}]$. For that purpose, we define the flux fields in terms of $\xi_x$ such that
\begin{equation*}
\begin{split}
&\tilde{\mathcal E}_{x-1,x}^n(t)-\tilde{\mathcal E}_{x,x+1}^n(t):=\xi_x(tn^{1+\alpha})-\xi_x(0) \\
\Big( \text{resp. \; }  &\tilde{\mathcal{V}}_{x-1,x}^n(t)-\tilde{\mathcal{V}}_{x,x+1}^n(t):=\log(\xi_x(tn^{1+\alpha}))-\log(\xi_x(0)).
\end{split}
\end{equation*}
As above, when it makes sense, we have that
\begin{equation*}
\begin{split}
&\tilde{\mathcal E}_{u_t^{x,\alpha}(n)}^n(t):=\sum_{y\geq{u_t^{x,\alpha}(n)}}\Big\{\xi_y(tn^{1+\alpha})-\xi_y(0)\Big\} \\
\Big( \text{resp. \; }  &\tilde{\mathcal{V}}_{u_t^{x,\alpha}(n)}^n(t):=\sum_{y\geq{u_t^{x,\alpha}(n)}}\Big\{\log(\xi_y(tn^{1+\alpha}))-\log(\xi_y(0))\Big\}\Big)
\end{split}
\end{equation*}
and in this case we can write the previous fields in terms of $\widehat Z_t^{n,\alpha}$.
A simple computation shows that Proposition \ref{prop nec} can similarly stated for last fields. Combining this with \eqref{eq:corr-jj} we have that
\begin{equation*}
\mathcal{E}_{u_t^{x,\alpha}(n)}^n(t):=\tilde{\mathcal E}_{u_t^{x,\alpha}(n)}^n(t)-\tilde{\mathcal V}_{u_t^{x,\alpha}(n)}^n(t),
 \quad \quad \mathcal{V}_{u_t^{x,\alpha}(n)}^n(t):=-\frac{1}{b}\tilde{\mathcal V}_{u_t^{x,\alpha}(n)}^n(t).
\end{equation*}
Then, applying  \eqref{vanish Z field}  to
$G_{\ell,x}^1(y)=
\left(
\begin{array}{c}
G_{\ell,x}(y) \\
0
\end{array}
\right)$  we obtain that
\begin{equation*}
\lim_{n \to \infty} {\mathbb E}_{\nu_{\beta,\lambda}} \left[ \left( \frac{1}{\sqrt n}\Big\{\tilde{\mathcal E}_{u_t^{x,\alpha}(n)}^n(t)-\bb E_{{\nu_{\beta,\lambda}}}[\tilde{\mathcal E}_{u_t^{x,\alpha}(n)}^n(t)]\Big\}\right)^2\right]=0.
\end{equation*}
On the other hand, applying \eqref{vanish Z field} to
\[
\tilde G_{\ell,x}(y)=
\left(
\begin{array}{c}
\frac{1}{\rho}G_{\ell,x}(y-u_t^{x,\alpha}(n)) \\
-G_{\ell,x}(y-u_t^{x,\alpha}(n))
\end{array}
\right)
\]
 we obtain that
\begin{equation*}
\begin{split}
\lim_{n \to \infty} {\mathbb E}_{\nu_{\beta,\lambda}} &\left[ \left( \frac{1}{\rho}\frac{1}{\sqrt n}\Big\{\tilde{\mathcal E}_{u_t^{x,\alpha}(n)}^n(t)-\bb E_{{\nu_{\beta,\lambda}}}[\tilde{\mathcal E}_{u_t^{x,\alpha}(n)}^n(t)]\Big\}\right.\right.\\
&\left.\left.\quad \quad \quad\quad \quad \quad\quad \quad \quad-\frac{1}{\sqrt n}\Big\{\tilde{\mathcal V}_{u_t^{x,\alpha}(n)}^n(t)-\bb E_{{\nu_{\beta,\lambda}}}[\tilde{\mathcal V}_{u_t^{x,\alpha}(n)}^n(t)] \Big\}\right)^2\right]=0.
\end{split}
\end{equation*}
Now, Corollary \ref{vanishing of Energy flux} follows easily from the previous results.
\begin{remark}\label{velocity}
From  \eqref{eq:hl-ss001}, the hydrodynamic equation of $\rho$ is independent of $\theta$ and  it can be rewritten as  $\partial_t \rho-2b^2\rho\partial_q\rho=0.$ Following the system along the characteristics for $\rho$, that is, removing the velocity $2b^2\rho$ from the system, we do not see a time evolution for $\rho$, and since $1/\rho\partial_t\rho-\partial_t\theta=0$,  nor for $\theta$. Therefore, translating the velocity $2b^2\rho$ in terms of the original variables it corresponds to $2b\bar\lambda/\bar\beta$ and that is the reason why we took the time dependent bond as written in Corollary \ref{vanishing of Energy flux}.
\end{remark}

%
%
%

\section{Diffusivity}
\label{sec:diff}

In this section we prove Theorem \ref{th:diffusivity}. Our proof is based on the resolvent methods introduced in \cite{B.,12} and developed in few other contexts (e.g. \cite{Bermerde,SS,TTV}). Some differences with these previous works are the presence of two and not only one conserved quantity and the degeneracy of the symmetric part of the generator.

{The main steps of the proof are the following. First  we use the microscopic change of variables and express the Laplace transform of the current-current correlation function as a resolvent norm in a suitable Hilbert space (see (\ref{eq:lapltransfW})). Then, we rewrite  this resolvent norm as the supremum over the set of local functions of a functional acting on these functions (see (\ref{eq:varformula11})). To get a lower bound we restrict the supremum over degree two  functions. The estimate of the value of the functional for a given degree two function remains in general very difficult. Thus we replace the functional restricted to the set of degree two functions by an equivalent functional simpler to estimate. This is accomplished through Lemma \ref{lem:compa}, Lemma \ref{lem:DF} and Lemma \ref{lem:007}. In the context of the asymmetric simple exclusion, this replacement step is called the ``free particles approximation'' (\cite{B.}) or the ``hard core removal'' (\cite{12}).  It is then possible to estimate the value of this equivalent functional for a suitable degree two test function.}

We fix $\rho>0, \theta \in \RR$ and denote by $\beta,\lambda$ the chemical potentials given by (\ref{eq:chimpot}). Let also $({\bar \beta}, {\bar \lambda})$ be given in terms of $(\beta,\lambda)$ by (\ref{eq:relblbl}).

Recall the definition of $\hat J_{x,x+1}$ given in \eqref{norm. currents}. We introduce the normalized currents ${\bj}_{x,x+1}$, ${\bj}^{\prime}_{x,x+1}$ and ${\bJ}_{x,x+1}$ corresponding to the process $(\xi (t))_{t \ge 0}$, which are defined by
\begin{equation*}
\begin{split}
&{\bj}_{x,x+1}(\xi) \\
&= j_{x,x+1}(\xi) - \langle j_{x,x+1} \rangle_{\nu_{\beta,\lambda}} - \partial_{\rho}  \langle j_{x,x+1} \rangle_{\nu_{\beta,\lambda}} (\xi_x -\rho) - \partial_{\theta}  \langle j_{x,x+1} \rangle_{\nu_{\beta,\lambda}} (\log (\xi_x) -\theta), \\
&{\bj}^{\prime}_{x,x+1}(\xi) \\
&= j^{\prime}_{x,x+1}(\xi) - \langle j^{\prime}_{x,x+1} \rangle_{\nu_{\beta,\lambda}} - \partial_{\rho}  \langle j^{\prime}_{x,x+1} \rangle_{\nu_{\beta,\lambda}} (\xi_x -\rho) - \partial_{\theta}  \langle j^{\prime}_{x,x+1} \rangle_{\nu_{\beta,\lambda}} (\log (\xi_x) -\theta),\\
&{\bJ}_{x,x+1}(\xi) = {\bj}_{x,x+1}(\xi) -{\bj}^{\prime}_{x,x+1}(\xi).
\end{split}
\end{equation*}

Since $\langle j_{x,x+1} \rangle_{\nu_{\beta,\lambda}}=-b^2 \rho^2$ and $\langle j^{\prime}_{x,x+1} \rangle_{\nu_{\beta,\lambda}} = -2b^2 \rho$, we get
\begin{equation}\label{norm. currents 2}
\begin{split}
&{\bj}_{x,x+1}(\xi) =-b^2 (\xi_{x} -\rho) (\xi_{x+1} -\rho) - (\gamma+b^2 \rho) \nabla \xi_x\\
&{\bj}^{\prime}_{x,x+1}(\xi) = -\nabla ( b^2 \xi_x + \gamma \log (\xi_x) ).
\end{split}
\end{equation}

For any local compactly supported functions $f,g:(0,+\infty)^{\ZZ}\rightarrow{\RR}$ we define the semi-inner product $\ll f, g \gg:=\ll f,g\gg_{\beta,\lambda}$ of $f$ and $g$ by
\begin{equation*}
\begin{split}
&\ll f , g \gg \\
&= \sum_{x \in \ZZ} \left( \langle \tau_x f g \rangle_{\nu_{\beta,\lambda}} -\langle f\rangle_{\nu_{\beta,\lambda}} \langle g \rangle_{\nu_{\beta,\lambda}} \right) \\
&= \lim_{k \to  \infty} \sum_{|x| \le k} \left( \langle \tau_x f g \rangle_{\nu_{\beta,\lambda}} -\langle f\rangle_{\nu_{\beta,\lambda}} \langle g \rangle_{\nu_{\beta,\lambda}} \right) \\
&= \lim_{k \to \infty} \frac{1}{2k+1} \sum_{|x| \le k}  \left\{ \sum_{|y-x| \le k} \Big( \langle \tau_{x+y} f \,\tau_y g \rangle_{\nu_{\beta,\lambda}} -\langle f\rangle_{\nu_{\beta,\lambda}} \langle g \rangle_{\nu_{\beta,\lambda}} \Big) \right\} \\
&= \lim_{k \to \infty}\!\! \left\langle \left( \cfrac{1}{\sqrt{2k+1}}\sum_{|x| \le k} (\tau_x f -\langle f\rangle_{\nu_{\beta,\lambda}})\right) \!\!\!\left( \cfrac{1}{\sqrt{2k+1}}\sum_{|x| \le k} (\tau_x g -\langle g\rangle_{\nu_{\beta,\lambda}})\right) \right\rangle_{\nu_{\beta,\lambda}}
\end{split}
\end{equation*}
where the third equality follows from the invariance of $\nu_{\beta,\lambda}$ by the shift. Observe also that the first sum on $\ZZ$ is in fact a finite sum since $f$ and $g$ are assumed to be local functions. We denote by ${\mc H}_0$ the space generated by the local compactly supported functions and the semi-inner product $\ll\cdot, \cdot \gg$. Observe that any constant  or gradient functions are equal to $0$ in ${\mc H}_0$.

By \eqref{norm. currents 2}, the normalized current associated to the volume is a gradient and this shows that ${\mc F}_{i,j} (\gamma,z)=0$ if $(i,j) \ne (1,1)$. By the definition of $\hat J_{x,x+1}$ and by \eqref{eq:corr-jj}, we are only interested in the behavior, as $z \to 0$, of
\begin{equation*}
{\mf L} (z) = \ll {\bJ}_{0,1}, (z -{\mc L})^{-1} {\bJ}_{0,1} \gg =\int_{0}^{\infty} \, e^{-z t} \, \ll {\bJ}_{0,1} (t) \, , \, {\bJ}_{0,1} (0) \gg \, dt.
\end{equation*}
Since gradient functions are equal to $0$ in ${\mc H}_0$, this is equivalent to estimate
\begin{equation}
\label{eq:lapltransfW}
{\mf L} (z) = b^4 \ll W_{0,1}, (z -{\mc L})^{-1} W_{0,1} \gg
\end{equation}
where $W_{x,y}$ is the local function $W_{x,y}=(\xi_{x} -\rho) (\xi_{y} -\rho)$.

In this section we prove that there exists a constant $C>0$ such that
\begin{equation}
\label{eq:z14}
\ll W_{0,1}, (z -{\mc L})^{-1} W_{0,1} \gg\,  \ge \, C z^{-1/4}.
\end{equation}

But before proving (\ref{eq:z14}) let us show (\ref{eq:F11c}) which is a direct consequence of the following lemma.

\begin{lemma}
For any $\gamma>0$, there exists a constant $C:=C(\gamma)$ such that
\begin{equation*}
\begin{split}
\ll W_{0,1}\, , \, (z/\gamma -b^2 {\mc A}-{\mc S})^{-1} W_{0,1} \gg \, \le \,  C \ll W_{0,1}, (z -b^2 {\mc A}-\gamma {\mc S})^{-1} W_{0,1} \gg
\end{split}
\end{equation*}
and
\begin{equation*}
 \ll W_{0,1}, (z -b^2 {\mc A}-\gamma {\mc S})^{-1} W_{0,1} \gg\,  \le \, C \ll W_{0,1}\, ,\,  (z/\gamma -b^2 {\mc A}- {\mc S})^{-1} W_{0,1} \gg.
\end{equation*}
\end{lemma}

\begin{proof}
Assume $\gamma>1$ the case $\gamma<1$ being similar. By Lemma 2.1 of \cite{B.} we have the variational formula for $\ll W_{0,1}, (z -\mc L)^{-1} W_{0,1} \gg$, where $\mathcal{L}=b^2\mc A+\gamma \mathcal S$, given by
\begin{equation*}
\sup_{f} \left\{ 2 \ll W_{0,1}, f \gg  - \ll f, (z -\gamma {\mc S}) f \gg - b^4 \ll {\mc A} f, (z -\gamma {\mc S})^{-1} {\mc A} f \gg \right\},
\end{equation*}
where the supremum is carried over functions $f$ belonging to the domain of the generator $\mathcal{L}$ or equivalently to a dense subspace included in this domain, say the space of smooth local compactly supported functions. We have that
\begin{equation*}
\begin{split}
&\sup_{f} \Big\{ 2 \ll W_{0,1}, f \gg  - \ll f, (z -\gamma {\mc S}) f \gg - b^4 \ll {\mc A} f, (z -\gamma {\mc S})^{-1} {\mc A} f \gg \Big\}\\
=&\sup_{f} \Big\{ 2 \ll W_{0,1}, f \gg  - \gamma \ll f, (z/\gamma - {\mc S}) f \gg - b^4 \gamma^{-1} \ll {\mc A} f, (z/\gamma - {\mc S})^{-1} {\mc A} f \gg \Big\}\\
\ge& \sup_{f} \Big\{ 2 \ll W_{0,1}, f \gg  - \gamma \ll f, (z/\gamma - {\mc S}) f \gg - b^4 \gamma \ll {\mc A} f, (z/\gamma - {\mc S})^{-1} {\mc A} f \gg \Big\}\\
=& \sup_{f}\Big\{ 2 \gamma^{-1/2}  \ll W_{0,1}, f \gg  -  \ll f, (z/\gamma - {\mc S}) f \gg - b^4 \ll {\mc A} f, (z/\gamma - {\mc S})^{-1} {\mc A} f \gg \Big\}
\end{split}
\end{equation*}
where the  inequality comes from $\gamma>1$ and last equality is obtained by the change of $f$ into $\gamma^{-1/2} f$. The last term is equal to
$$ \gamma^{-1} \ll W_{0,1}\, ,\,  (z/\gamma  -b^2 {\mc A}- {\mc S})^{-1} W_{0,1} \gg$$
and this proves the first inequality of the lemma.

 For the second one we proceed similarly:
\begin{equation*}
\begin{split}
&\sup_{f} \Big\{ 2 \ll W_{0,1}, f \gg  - \ll f, (z -\gamma {\mc S}) f \gg - b^4 \ll {\mc A} f, (z -\gamma {\mc S})^{-1} {\mc A} f \gg \Big\}\\
=&\sup_{f} \Big\{ 2 \ll W_{0,1}, f \gg  - \gamma \ll f, (z/\gamma - {\mc S}) f \gg - b^4 \gamma^{-1} \ll {\mc A} f, (z/\gamma - {\mc S})^{-1} {\mc A} f \gg \Big\}\\
\le& \sup_{f} \Big\{ 2 \ll W_{0,1}, f \gg  - \gamma^{-1} \ll f, (z/\gamma - {\mc S}) f \gg - b^4 \gamma^{-1} \ll {\mc A} f, (z/\gamma - {\mc S})^{-1} {\mc A} f \gg \Big\}\\
=&\sup_{f} \Big\{ 2 \gamma^{1/2}  \ll W_{0,1}, f \gg  -  \ll f, (z/\gamma - {\mc S}) f \gg - b^4 \ll {\mc A} f, (z/\gamma - {\mc S})^{-1} {\mc A} f \gg \Big\}\\
=&\gamma \ll W_{0,1}\, ,\,  (z/\gamma  -b^2 {\mc A}- {\mc S})^{-1} W_{0,1}\gg.
\end{split}
\end{equation*}
\end{proof}

Recall the orthogonal decomposition described in Section \ref{sec:dua}. Let $f=\sum_{\sigma} F(\sigma) H_{\sigma}$ and $g= \sum_{\sigma} G(\sigma) H_{\sigma}$ be two centered local functions. The configuration $\sigma$ shifted by $z \in \ZZ$ is denoted by $\tau_z \sigma$, that is $\tau_z\sigma(x)=\sigma(x-z)$. We identify $F_n, G_n$, the restrictions of $F,G$ to $\Sigma_n$, with symmetric functions on $\ZZ^n$.  By \eqref{eq:prod Hsigma} we have that
\begin{equation*}
\ll f , g \gg = \sum_{z\in{\ZZ}}\sum_{\sigma \in \Sigma} F (\tau_z \sigma) G (\sigma) {\mc W} (\sigma),
\end{equation*}
where ${\mc W}$ was defined in \eqref{def W}.

With some abuse of notations, we denote by $\ll F, G \gg$ the scalar product defined by
\begin{equation*}
\ll F, G \gg= \sum_{z\in{\ZZ}}\sum_{\sigma \in \Sigma} F (\tau_z \sigma) G (\sigma) {\mc W} (\sigma).
\end{equation*}

We also introduce the inner product $\fsp{\cdot, \cdot}$ defined by
\begin{equation*}
\begin{split}
&\fsp{F, G} \;= \;  \sum_{ y \in \ZZ}\sum_{\sigma \in \Sigma} F(\tau_y \sigma) G(\sigma).\\
\end{split}
\end{equation*}

Since the function ${\mc W}$ is invariant by the shift, we have a very simple relation between these two inner products:
\begin{equation}
\label{eq:fspsp}
\ll F, G \gg \; =\; \fsp{{\mc W}^{1/2} F, {\mc W}^{1/2} G}.
\end{equation}

On the set $\Sigma_n$ we introduce the equivalence relation $\star$ defined by $\sigma \star \sigma'$ if and only if there exists $u \in \ZZ$ such that $\tau_u \sigma =\sigma'$. Let $\Sigma_n^\star = \Sigma_n/ \star$ be the set of classes for this relation and $\Sigma^\star = \cup_{n \ge 1} \Sigma_n^{\star}$.  We can rewrite the scalar product $\fsp{\cdot, \cdot}$ as
\begin{equation*}
\fsp{F,G} = \sum_{{\bar \sigma} \in \Sigma^\star} {\bar F} ({\bar \sigma}) {\bar G} ({\bar \sigma}).
\end{equation*}
Here $\bar F$ is defined by $\bar F ({\bar \sigma})=\left(\sum_{y \in \ZZ} \tau_y F\right) (\sigma)$ where $\sigma$ is any element of ${\bar \sigma}$. The function ${\mc W}$ being invariant by the shift, we define ${\mc W} (\bar \sigma)$ by $\mc W (\sigma)$, $\sigma \in {\bar \sigma}$, ${\bar \sigma} \in {\Sigma}^{\star}$. Then, we have
\begin{equation*}
\ll F , G \gg = \sum_{{\bar \sigma} \in \Sigma^\star} {\mc W} ({\bar \sigma}) {\bar F} ({\bar \sigma}) {\bar G} ({\bar \sigma}).
\end{equation*}

\begin{lemma}
\label{lem:compa}
There exists a constant $C:=C(n,\lambda)$ such that for any local function $F:\Sigma_n \to \RR$ of degree $n$ it holds that
\begin{enumerate}

\item \begin{equation*}
C^{-1} \fsp{F, F} \; \le\;  \ll F, F \gg \;\le \; C \fsp{F, F}.
\end{equation*}

\item $$C^{-1} \fsp{F, -{\mf S} F} \; \le\;  \ll F, - {\mf S} F \gg \;\le \; C \fsp{F, -{\mf S} F}. $$
\end{enumerate}
Moreover, for  any positive real $z>0$
\begin{equation*}
\begin{split}
 \ll F, (z- \gamma {\mf S})^{-1} F \gg =\fsp{{\mc W}^{1/2} F\, ,\,  (z- \gamma {\mf S})^{-1} \, {\mc W}^{1/2} F}.
\end{split}
\end{equation*}


\end{lemma}

\begin{proof}
Recall the definition of ${\mc W}$ from \eqref{def W}.
 Thus,  ${\mc W}$ is bounded from above (resp. from bellow) by a constant $C(n,\lambda)$ (resp. $C^{-1} (n,\lambda)$) independent of $\sigma \in \Sigma_n$. This is enough to conclude $\textit{(1)}$. In order to prove \textit{(2)}, it is enough to use \eqref{eq:fspsp} and the fact that for any local function $F: \Sigma \to \RR$ we have that ${\mf S} ({\mc W}^{1/2} F)= {\mc W}^{1/2} {\mf S} F$.
Finally, for a local function $F$ of degree $n$, we have by (\ref{eq:fspsp}) and the fact that
\begin{equation*}
\ll F, (z- \gamma {\mf S})^{-1} F \gg = \sup_{G {\text{ of degree $n$}} } \left\{  2 \ll F, G \gg - \ll G, (z- \gamma {\mf S}) G \gg \right\},
\end{equation*}
the following equality
\begin{equation*}
\begin{split}
 \ll F, (z- \gamma {\mf S})^{-1} F \gg =\fsp{{\mc W}^{1/2} F, (z- \gamma {\mf S})^{-1} {\mc W}^{1/2} F},
\end{split}
\end{equation*}
which proves the last assertion.

\end{proof}

Our goal is to get a lower bound for $\ll W_{0,1}, (z -{\mc L})^{-1} W_{0,1} \gg$ which by Lemma 2.1 of \cite{B.} can be rewritten in the variational form
\begin{equation}
\label{eq:varformula11}
\sup_{f} \Big\{ 2 \ll W_{0,1}, f \gg  - \ll f, (z -\gamma {\mc S}) f \gg - b^4 \ll {\mc A} f, (z -\gamma {\mc S})^{-1} {\mc A} f \gg \Big\}.
\end{equation}

Any element $\bar \sigma \in \Sigma_n^{\star}$ can be identified with an element of ${\mathbb N}^{n-1}$ through the application which associates to $(\alpha_1, \ldots, \alpha_{n-1}) \in {\NN}^{n-1}$ the class of the configuration $\sigma =\delta_{0} + \delta_{\alpha_1} + \ldots +\delta_{\alpha_1 + \ldots+ \alpha_{n-1}}$.

 Observe also that ${\mf S}$ is a self-adjoint operator with respect to $\ll \cdot, \cdot \gg$ and with respect to $\fsp{\cdot,\cdot}$. We restrict the previous supremum over degree $2$ functions $f=\sum_{(x,y) \in \ZZ^2} F([x,y]) H_{[x,y]}$.
In order to keep notation simple, whenever we identify a configuration $\sigma\in\Sigma_n$ with $[\textbf{x}]\in \mathbb{Z}^n$ we will simply write $F(x)$, instead of $F([\textbf{x}])$.

Up to some irrelevant multiplicative constant, a lower bound is given by
\begin{equation*}
\sup_{F {\text{of degree $2$}} } \Big\{ 2 F(0,1) - \| F \|_{1,z}^2 -b^4 \| {\mf A}_- F \|_{-1,z}^2 -b^4 \| {\mf A}_+ F \|_{-1,z}^2 -b^4\| {\mf A}_0 F \|_{-1,z}^2\Big\}
\end{equation*}
where $\| F \|_{\pm 1, z}^2 = \ll F, (z- \gamma {\mf S})^{\pm 1} F \gg$. We also introduce the corresponding $H_{{\pm 1,z}}$-norms associated to $\fsp{\cdot,\cdot}$:
$\| F \|_{\pm 1, z,{\rm{free}}}^2 = \fsp{ F, (z- \gamma {\mf S})^{\pm 1} F }$, for $F:\Sigma\rightarrow{\bb R}$.

By Lemma \ref{lem:compa}, there exists a constant $C$ such that this lower bound is bounded from bellow by
\begin{equation*}
\begin{split}
&\sup_{F {\text{of degree $2$}} } \left\{  2 F(0,1) - C \| F \|_{+1,z,{\rm{free}}}^2\right. \\
& \left.  -b^4 \| {\mc W}^{1/2} {\mf A}_- F \|_{-1,z, {\rm free}}^2 -b^4 \|{\mc W}^{1/2}  {\mf A}_+ F \|_{-1,z, {\rm free}}^2 -b^4\| {\mc W}^{1/2} {\mf A}_0 F \|_{-1,z, {\rm free}}^2\right\}.
\end{split}
\end{equation*}

Let us first show that if $F$ is of degree $2$ then the contributions given by $ \| {\mc W}^{1/2} {\mf A}_- F \|_{-1,z, {\rm free} }^2$ and $\| {\mc W}^{1/2}  {\mf A}_0 F \|_{-1,z, {\rm free} }^2$ are equal to zero.

The function ${\mc W}$ is constant and equal to $(\lambda+1)$
 on $\Sigma_1$ so that ${\mc W}^{1/2} {\mf A}_- F =\sqrt{\lambda+1} {\mf A}_- F $. It is easy to check that the degree one function ${\mf A}_- F$ satisfies
\begin{equation*}
({\mf A}_- F)(u)=(\lambda-1)\Big( F(u-1,u) - F(u,u+1)\Big).
\end{equation*}

For any degree $1$ function $G$, we have
\begin{equation*}
\fsp{ {\mf A}_- F , G }= \sum_{u,y \in \ZZ} G(u+y)(\lambda-1)\Big( F(u-1,u) - F(u,u+1)\Big) =0
\end{equation*}
by a telescopic sum argument. This shows that ${\mf A}_- F$ is equal to zero in the Hilbert space generated by $\fsp{\cdot, \cdot}$.


Recall that if $F$ is a degree $2$ function, i.e. a symmetric function on $\ZZ^2$, then ${F}$ is identified with a function $\bar{F}$ defined on $\NN$ by
\begin{equation*}
{\bar F} (\alpha) = \sum_{u \in \ZZ} F(u,u +\alpha)
\end{equation*}
and as a consequence, for $F$ and $G$ degree $2$ functions it holds that
\begin{equation}\label{eq: eqfree}
\fsp{F,G}\; =\; \sum_{\alpha \in \NN} {\bar F} (\alpha) {\bar G} (\alpha).
\end{equation}

Observe that  $({\mf A}_0 F)(u,v)$ is equal to
\begin{equation*}
\begin{split}
\begin{cases}
2(1+\lambda) \Big( F(u-1,u) -F(u,u+1)\Big),\quad {\text{if}} \; u=v,\\
(1+\lambda)  \Big( F(u-1,u+1) -F(u,u+2)\Big) +(2+\lambda)\Big(F(u,u) -F(u+1,u+1) \Big),\\
\quad \quad {\text{if}} \; (u,v)=(u,u+1),\\
(1+\lambda)\Big( F(u-1,v) - F(u+1,v) +F(u,v-1)-F(u,v+1) \Big), \quad {\text {if}} \; |u-v| \ge 2
\end{cases}
\end{split}
\end{equation*}
and
\begin{equation}\label{expression for W}
{\mc W} (u,u)= \cfrac{(\lambda+1)(\lambda+2)}{2}, \quad {\mc W} (u,v)= (\lambda+1)^2\quad  \text{ for } u \ne v.
\end{equation}

It is then easy to show that
\begin{equation*}
\overline{{\mc W}^{1/2} ({\mf A}_0 F)} (\alpha)=0
\end{equation*}
for any $\alpha \in \NN$. Putting together the previous result and \eqref{eq: eqfree} it follows that:

\begin{equation*}
\begin{split}
\|{\mc W}^{1/2}  {\mf A}_0 F\|_{-1,z, {\rm{free}}}^2&= \fsp{  {\mc W}^{1/2} {\mf A}_0 F\; ,\;  (z - \gamma {\mf S})^{-1} ({\mc W}^{1/2} {\mf A}_0 F) }\\
&= \sum_{\alpha \in \NN} \overline{{\mc W}^{1/2} ({\mf A}_0 F)} (\alpha) \; \overline{(\lambda -\gamma {\mf S})^{-1} [{\mc W^{1/2} ({\mf A}_0 F)]}} (\alpha) =0.
\end{split}
\end{equation*}

\begin{lemma}
\label{lem:DF}
There exists a positive constant $C$ such that for every symmetric function $F$ of degree $2$, if ${\bar F} (\alpha) = \sum_{z\in \bb Z} F(z,z+\alpha)$, then
\begin{equation*}
C^{-1}  \sum_{\substack{x,y \ne 0,\\ |x-y|=1}} \Big( {\bar F} (y) -{\bar F} (x) \Big)^2 \; \le \; \fsp{F, -{\mf S} F }\; \le \; C \sum_{\substack{x,y \ne 0,\\ |x-y|=1}} \Big({\bar F} (y) -{\bar F} (x) \Big)^2.
\end{equation*}
\end{lemma}

\begin{proof}
This follows easily from the following equalities together with \eqref{eq: eqfree}:
\begin{equation*}
\begin{split}
\overline{{\mf S} F} (0) &=\sum_{y\in{\ZZ}} ({\mf S} F) (y,y)\\
&= \sum_{y\in{\ZZ}} \Big( F (y+1,y+1) -F(y,y) \Big) +\Big( F (y-1,y-1) -F(y,y) \Big)=0,
\end{split}
\end{equation*}
\begin{equation*}
\begin{split}
\overline{{\mf S} F} (1) &= \sum_{y\in{\ZZ}} ({\mf S} F) (y,y+1)\\
&= \sum_{y\in{\ZZ}} \Big( F(y-1,y+1) -F(y,y+1)\Big) +  \sum_{y\in{\ZZ}} \Big( F(y,y+2) -F(y,y+1)\Big)\\
&=2 \Big({\bar F} (2) - {\bar F} (1)\Big),
\end{split}
\end{equation*}
\begin{equation*}
\begin{split}
\overline{{\mf S} F} (\alpha)&= \sum_{y} ({\mf S} F) (y,y+\alpha)\\
&= \sum_{y\in{\ZZ}} \Big( F(y-1,y+\alpha) -F(y,y+\alpha)\Big) +  \sum_{y\in{\ZZ}}\Big( F(y+1,y+\alpha) -F(y,y+\alpha)\Big)\\
&+ \sum_{y\in{\ZZ}} \Big( F(y,y+\alpha+1) -F(y,y+\alpha)\Big) +  \sum_{y\in{\ZZ}} \Big( F(y,y+\alpha-1) -F(y,y+\alpha)\Big)\\
&= 2 \Big( {\bar F} (\alpha + 1) -{\bar F} (\alpha)\Big)+2 \Big( {\bar F} (\alpha - 1) -{\bar F} (\alpha)\Big), \quad \alpha \ge 2.
\end{split}
\end{equation*}
\end{proof}


To any degree $3$ function $G$, i.e. a symmetric function $G$ on $\ZZ^3$, the function ${\bar G}$ is identified with a function on $\NN^2$:
\begin{equation*}
{\bar G} (u,v) =\sum_{y \in \ZZ} G(y,u+y,u+v+y).
\end{equation*}

Since G is symmetric on $\ZZ^3$, then $\bar G$ is symmetric on $\ZZ^2$. As above, for $F$ and $G$ degree $3$ functions it holds that
\begin{equation}\label{eq: eqfree2}
\fsp{F,G}\; =\; \sum_{(\alpha,\beta) \in \NN^2} {\bar F} (\alpha,\beta) {\bar G} (\alpha,\beta).
\end{equation}

Let ${\bb D}_3$, acting on the local functions on $\NN^2$, be defined by
\begin{equation}
\begin{split}
{\bb D}_3 (\bar G) &= \sum_{u \ge 1} \Big( {\bar G} (u+1,0) -{\bar G} (u,0) \Big)^2 + \sum_{v \ge 1} \Big( {\bar G} (0,v+1) -{\bar G} (0,v) \Big)^2 \\&
+ \Big({\bar G} (1,0) -{\bar G} (0,1)\Big)^2+ \sum_{u,v \ge 1} \Big( {\bar G} (u+1,v) -{\bar G} (u,v)\Big)^2 \\
&+  \Big( {\bar G} (u,v+1) -{\bar G} (u,v)\Big)^2.
\end{split}
\end{equation}
This is the Dirichlet form of a symmetric nearest neighbors random walk on $\NN^2$ where all the jumps between $\{0\} \times \NN$ and $\NN^*\times \NN^*\footnote{Here and in the sequel $\NN^*:=\NN\backslash{\{0\}}$}$, all the jumps from $\NN \times \{0\}$ and $\NN^*\times \NN^*$ and all the jumps from $0$ have been suppressed, and a jump between $(0,1)$ and $(1,0)$ has been added.

\begin{lemma}
\label{lem:007}
There exists a constant $C>0$ such that for any symmetric function $G$ on $\ZZ^3$
\begin{equation*}
C^{-1} {\bb D}_3 ({\bar G}) \; \le \; \fsp{ G, -{\mf S} G } \; \le \; C  {\mathbb D}_3 ({\bar G}).
\end{equation*}
\end{lemma}

\begin{proof}
We have the following equalities
\begin{equation*}
\begin{split}
&\overline{{\mf S} G } \, (0,0)=0,\\
&\overline{{\mf S} G } \, (0,1)= 2 \Big( {\bar G} (0,2) -{\bar G} (0,1) \Big) + \Big({\bar G} (1,0) -{\bar G} (0,1)\Big),\\
&\overline{{\mf S} G } \, (0,\beta)= 2 \Big( {\bar G} (0,\beta + 1) -{\bar G} (0,\beta) \Big)+2\Big( {\bar G} (0,\beta - 1) -{\bar G} (0,\beta)\Big), \quad \beta \ge 2,\\
&\overline{{\mf S} G } \, (1,0)= 2 \Big( {\bar G} (2,0) -{\bar G} (1,0) \Big) + \Big({\bar G} (0,1) -{\bar G} (1,0)\Big)\\
&\overline{{\mf S} G } \, (\alpha,0)= 2 \Big( {\bar G} (\alpha + 1) -{\bar G} (\alpha,0) \Big)+2 \Big({\bar G} (\alpha - 1) -{\bar G} (\alpha,0) \Big), \quad \alpha \ge 2,\\
&\overline{{\mf S} G } \, (\alpha,\beta) = \Big( {\bar G} (\alpha+1, \beta) -{\bar G}(\alpha,\beta) \Big) + \Big( {\bar G} (\alpha, \beta+1) -{\bar G}(\alpha,\beta) \Big)\\
&+ {\bf 1}_{\{\alpha \ge 2\}}\Big( {\bar G} (\alpha-1, \beta +1) -{\bar G} (\alpha,\beta)\Big) + {\bf 1}_{\{\alpha \ge 2\}} \Big({\bar G} (\alpha-1, \beta) -{\bar G} (\alpha,\beta)\Big) \\
&+{\bf 1}_{\{\beta \ge 2\}}\Big({\bar G} (\alpha+1, \beta -1) -{\bar G} (\alpha,\beta)\Big) + {\bf 1}_{\{\beta \ge 2\}}\Big({\bar G} (\alpha, \beta-1) -{\bar G} (\alpha,\beta)\Big), \quad \alpha,\beta \ge 1.
\end{split}
\end{equation*}

We recognize in these expressions the generator of a symmetric nearest neighbors random walk on $\NN^2$ where
\begin{itemize}
\item all the jumps between $\{0\} \times \NN$ and ${\NN}^*\times \NN^*$, all the jumps between $\NN \times \{0\}$ and ${\NN}^*\times \NN^*$, and all the jumps from $0$ have been suppressed;
\item a jump between $(0,1)$ and $(1,0)$ with rate $1$ has been added;
\item jumps between $(\alpha,\beta)$ and $(\alpha \pm 1, \beta \mp 1)$ for $(\alpha,\beta) \in \NN^*\times\NN^*$ with rate $1$ have been added.
\item the non vanishing jumps on $\NN \times \{0\}$ and on $\{0\} \times \NN$ have been multiplied by $2$.
\end{itemize}
 This together with \eqref{eq: eqfree2}, implies the lemma.
\end{proof}

We choose a degree $2$ symmetric function $F$ such that
\begin{equation}
\label{eq:FF}
\begin{split}
&{\overline F} (\alpha) = z^{-1/4} e^{-z^{3/4} (\alpha -1)}, \quad \alpha \ge 1,\\
&{\overline F} (0) ={\bar F} (1).
\end{split}
\end{equation}

This function exists since given a function  $G$ defined on $\NN$ we can find a symmetric function $F$ defined in $\ZZ^2$ such that $\bar F=G$. For that purpose, take $F(x,y) ={G} (|y-x|) [ \phi (x) + \phi (y)]$ where the function $\phi$ is defined on $\ZZ$ and is such that $\sum_{x\in{\ZZ}} \phi (x) =1/2$. Then for any $\alpha \in \NN$, ${\bar F} (\alpha) = \sum_{u \in \ZZ} F(u,u+\alpha)= G(\alpha) \sum_{u \in \ZZ} [\phi (u) + \phi (u+\alpha)]= G(\alpha)$.

Observe that with this choice, by Lemma \ref{lem:DF},
\begin{equation}
\label{eq:eraclio}
\fsp{ F, -{\mf S} F }\sim z^{1/4}, \quad {\bar F} (1) = z^{-1/4}, \quad z \sum_{\alpha \in \NN} {\bar F}^2 (\alpha) \sim z^{-1/4}.
\end{equation}

It remains to estimate the last contribution given by $\| {\mc W}^{1/2}  G\|_{-1,z,{\rm{free}}}^2$ where $G={\mf A}_{+} F $ is a degree $3$ function.

\begin{lemma}
\label{lem:hadrien}
Let $G={\mf A}_{+} F $ where $F$ is defined by (\ref{eq:FF}). There exists a constant $C>0$ such that
$$\| {\mc W}^{1/2}  G\|_{-1,z,{\rm{free}}}^2 \ge C z^{-1/4}.$$
\end{lemma}

\begin{proof}
For any $u,v,w \in \ZZ$, we have
\begin{equation*}
\begin{split}
& G(u,u+1,u+2)=  F(u, u+1)- F(u+1,u+2) ,\\
& G(u,u+1,v)  = F(u,v)-F(u+1,v), \quad v>u+1,\\
&G(v,u,u+1)=  F(v,u)-F(v,u+1), \quad v<u,\\
& G(u,u,u+1)= 2\Big( F(u,u) - F(u,u+1)\Big),\\
& G(u,u,u-1) = 2 \Big(F(u-1,u) - F(u,u) \Big),\\
&G(u,v,w)=0 \quad {\text{otherwise.}}
\end{split}
\end{equation*}

Let us now compute ${\bar G}(u,v)$, $u,v \in \NN$. We get
\begin{equation}
\label{eq:GGG}
\begin{split}
&{\bar G} (0,1)=-{\bar G} (1,0)=2 {\bar F} (0) -2 {\bar F} (1),\\
&{\bar G} (1,v)= {\bar F} (v+1) -{\bar F} (v), \quad v \ge 2,\\
&{\bar G} (u,1)=  {\bar F} (u) -{\bar F} (u+1) , \quad u \ge 2,\\
&{\bar G} (u,v)=0 \quad {\text{otherwise.}}
\end{split}
\end{equation}
By (\ref{eq:FF}) we have that ${\bar G} (0,1)= {\bar G} (1,0)=0$. Also notice that $\bar G(u,u)=0$ and by \eqref{expression for W} we have that ${\mc W}^{1/2} (u,v)=(1+\lambda)$ for $u\neq{v}$.

It follows, by Lemma {\ref{lem:007}}, that $\|{\mc W}^{1/2} G \|_{-1,z,{\rm{free}}}^2$ is upper bounded by the variational formula:

\begin{equation*}
\begin{split}
\| {\mc W}^{1/2} G \|_{-1,z, {\rm{free}}}^2 &= \sup_{R} \left\{ 2 \sum_{(u,v){\in{\bb N^2}}} R(u,v) \mc{W}^{1/2}(u,v){\bar G} (u,v) -C_{0} {\bb D}_3 (R) \right\}\\
&= \sup_{R} \left\{ 2 (1+\lambda)\sum_{(u,v)\in{\bb N^2}} R(u,v) {\bar G} (u,v) -C_{0} {\bb D}_3 (R) \right\}
\end{split}
\end{equation*}
where the supremum is taken over local functions on $\NN^2$. By (\ref{eq:GGG}), we have that
\begin{equation}
\label{eq:trajan}
\begin{split}
&\sum_{(u,v)\in{\bb N^2}} R(u,v) {\bar G} (u,v)\\&
=  \sum_{v \ge 2} R (1,v)\Big( {\bar F} (v+1) -{\bar F} (v) \Big)- \sum_{u \ge 2} R(u,1) \Big( {\bar F} (u+1) - {\bar F} (u) \Big)\\
&=  \sum_{v \ge 3} {\bar F} (v) \Big( R(1,v-1) -R(1,v)\Big)-  \sum_{u \ge 3} {\bar F} (u) \Big( R(u-1,1) -R(u,1)\Big) \\
&+  {\bar F} (2)\Big(R(2,1) -R(1,2) \Big)\\
&=   \sum_{v \ge 2} {\bar F} (v) \Big( R(1,v-1) -R(1,v)\Big)-\sum_{u \ge 2} {\bar F} (u) \Big( R(u-1,1) -R(u,1)\Big).
\end{split}
\end{equation}

We use now the following parametrization of $R$. For $k \ge 1$, $v \in \ZZ$, let us define
\begin{equation*}
R(k,v)=\phi (k-1,v-k), \quad v \ge k, \quad R(u,k)=\phi (k-1,-u+k), u \ge k,
\end{equation*}
where $\{\phi (k,\cdot) \; ;\; k \ge 0\}$ are functions from $\ZZ \to \RR$. We have the following lower bound for ${\bb D}_3 (R)$:
\begin{equation*}
{\bb D}_3 (R) \ge \sum_{u,v \ge 1} \Big( R (u+1,v) -R (u,v)\Big)^2 +  \Big( R (u,v+1) -R (u,v)\Big)^2
\end{equation*}
which is nothing but the Dirichlet form of a random walk where only jumps connecting sites of $\NN^* \times \NN^*$ have been conserved. With the choice of the parametrization for $R$ and this lower bound, it is not difficult to show there exists a constant $C>0$ such that
\begin{equation*}
{\bb D}_3 (R) \ge C \sum_{k \ge 0} \sum_{v \in \ZZ} \Big( \phi (k, v+1) -\phi (k,v) \Big)^2 +\Big( \phi (k+1, v) -\phi (k,v) \Big)^2.
\end{equation*}
The right hand side of the previous inequality is the Dirichlet form of a symmetric simple random walk on $\NN \times \ZZ$.

By (\ref{eq:trajan}), we get
\begin{equation*}
\sum_{(u,v)\in \bb N^2} R(u,v) {\bar G} (u,v)=
\sum_{u \in \ZZ} \phi (0,u) \Big( {\tilde F}(u-1) -{\tilde F} (u) \Big)
\end{equation*}
where ${\tilde F}: \ZZ \to \RR$ is defined by ${\tilde F} (u)= -{\bar F} (u+2) {\bf 1}_{\{u \ge 0\}} -{\bar F} (1-u) {\bf 1}_{\{u \le -1\}}$. We extend the function $\phi$ defined on $\NN \times \ZZ$ to $\ZZ^2$ by defining $\phi(-k, u)= \phi (k,u)$, $k \ge 1, u \in \ZZ$. Observe then that
\begin{equation*}
\begin{split}
{\bb D}_3 (R) &\ge C \sum_{k \ge 0} \sum_{v \in \ZZ}\Big( \phi (k, v+1) -\phi (k,v) \Big)^2 + \Big( \phi (k+1, v) -\phi (k,v)\Big)^2 \\
&= \cfrac{C}{2} \sum_{k \in \ZZ} \sum_{v \in \ZZ} \Big( \phi (k, v+1) -\phi (k,v) \Big)^2 + \Big( \phi (k+1, v) -\phi (k,v) \Big)^2. \\
\end{split}
\end{equation*}

Consequently we have, for suitable positive constants $C_1, C_2$:
\begin{equation}
\label{eq:flavien}
\begin{split}
\| {\mc W}^{1/2} G \|_{-1,z, {\rm{free}}}^2 \le C_1 \sup_{\phi} \Big\{ 2 \sum_{u \in \ZZ} \phi(0,u)& \Big( {\tilde F} (u-1) -{\tilde F} (u) \Big)\\
&-C_{2} \sum_{\substack{(\bu,\bv)\in \ZZ^2\\ |\bu-\bv|=1}} \Big( \phi (\bu) -\phi (\bv) \Big)^2\Big \}.
\end{split}
\end{equation}

A standard Fourier computation shows this supremum is of order $z^{-1/4}$.  Indeed, let ${\widehat u}$ be the Fourier transform of the function $u: \ZZ^n \to \RR$, defined by
\begin{equation*}
{\widehat u} (\bk) = \sum_{\bx \in \ZZ^n} e^{2i \pi \bx \cdot \bk} u (\bx), \quad \bk=(k_1, \ldots,k_n),
\end{equation*}
and denote by ${\widehat u}^* (\bk)$ the complex conjugate of ${\widehat u} (\bk)$. Using the expression of the sum of a convergent geometric series, we obtain the following expression for the Fourier transform $\Psi (k_1, k_2)$ of the function $(x,y) \in \ZZ^2 \to \delta_0 (y) {\tilde F} (x)$:
\begin{equation*}
\Psi (k_1, k_2) = - z^{-1/4} e^{-z^{3/4}} \left\{ \cfrac{1}{1- e^{2i\pi k_1} e^{-z^{3/4}}}-\cfrac{e^{-2i \pi k_1}} {1-e^{-2i\pi k_1} e^{-z^{3/4}} } \right\}
\end{equation*}
which satisfies
\begin{equation*}
\left| \Psi (k_1, k_2) \right| \le \cfrac{C_3 \sqrt{z}}{z^{3/2} + C_4 \sin^{2} (\pi k_1) }
\end{equation*}
for some positive constants $C_3, C_4$. The supremum appearing in (\ref{eq:flavien}) is then given by
\begin{equation*}
C_2^{-1} \int_{[0,1]^2} \cfrac{| \Psi (k_1, k_2)|^2 }{z+ 4 \sin^2 (\pi k_1) + 4 \sin^{2} (\pi k_2) } dk_1 dk_2.
\end{equation*}
Then the result follows by a standard study of this integral.
\end{proof}

To obtain (\ref{eq:z14}), by (\ref{eq:eraclio}) and Lemma \ref{lem:hadrien},  it suffices to take a test function in the form $aF$ with $F$ given by (\ref{eq:FF}) and $a$ sufficiently small.

\section{Stochastic perturbations of Hamiltonian systems}
\label{sec:pert}

In this section we discuss some other possible stochastic perturbations and make some connections with the recent models considered in \cite{BC}. Let us start with the Hamiltonian system (\ref{eq:dyneq}) with potential $V$ and generator $A$ given by
\begin{equation*}
A =\sum_{x\in{\ZZ}} \Big(V' (\eta_{x+1}) -V' (\eta_{x-1}) \Big) \partial_{\eta_x}.
\end{equation*}
The energy $\sum_{x\in{\ZZ}} V(\eta_x)$ and the volume $\sum_{x\in{\ZZ}} \eta_x$ are conserved by these dynamics. Remark that in fact $\sum_{x\in{\ZZ}} \eta_{2x}$ and $\sum_{x\in{\ZZ}} \eta_{2x+1}$ are also conserved and that we cannot exclude the case that still many others exist. This is the case for example for the exponential interaction for which an infinite number of conserved quantities can be explicitly identified. Anyway, we are only interested in these two first quantities. The product probability measures $\mu_{\beta,\lambda}$ defined by
\begin{equation*}
\mu_{\beta,\lambda} (d\eta) = \prod_{x \in \ZZ} Z(\beta,\lambda)^{-1}
\exp\left\{ -\beta V( \eta_x) -\lambda \eta_x \right\} \, d\eta_x, \\
\end{equation*}
where
\begin{equation*}
Z(\beta, \lambda) = \int_{-\infty}^{+\infty} \exp\left( -\beta V(r) -\lambda r \right)\, dr.
\end{equation*}
are invariant for the infinite dynamics.

In \cite{BS} we proposed to perturb this deterministic dynamics by the Poissonian noise considered in this paper and conserving both the energy and the volume. One could also consider the `` Brownian'' noise whose generator $S$ is given by $ S= \sum_{x\in{\ZZ}} Y_x^2$ where
$$Y_x\!=\!(V'(\eta_{x+1})-V'(\eta_{x-1}))\! \partial_{\eta_{x}}\!+(V'(\eta_{x-1})-V'(\eta_{x}))\! \partial_{\eta_{x+1}}\! + (V'(\eta_{x})-V'(\eta_{x+1})) \! \partial_{\eta_{x-1}},$$
is the vector field tangent to the curve
$$\Big\{ (\eta_{x-1}, \eta_x, \eta_{x+1}) \in \RR^3 \, ; \sum_{y=x-1}^{x+1} \eta_y =0, \; \sum_{y=x-1}^{x+1} V (\eta_y) =1\Big\}.$$
It is easy to see that the process with generator $L=A+S$ conserves the energy and the volume and has $\mu_{\beta, \lambda}$ as invariant measures. A priori, it should be possible to extend our result to this system for $V$ of exponential type but the noise $S$ seems to have a quite complicated expression in the orthogonal basis we used in this paper.  The advantage of the Poissonian noise is its very simple form. Notice also that the Poissonian noise is a weaker perturbation of the Hamiltonian dynamics than the Brownian noise in the sense it is less mixing. Indeed, consider the discrete torus ${\bb T}^N$ of length $N$ and the Brownian noise ${S}^N=\sum_{x\in \TT^N} Y_x^2$ restricted to the manifold ${\mc M}_{\pi, E}^N$ defined by
$${\mc M}^N_{\pi,E} = \left\{ \eta \in \RR^{\TT^N} \, ; \, \sum_{y \in \TT^N} \eta_y =\pi, \; \; \sum_{y \in \TT^N} V(\eta_y) =E \right\}, \quad E>0, \pi \in \RR.$$
Then $S^N$ is ergodic on ${\mc M}_{\pi, E}^N$ but this is not true for the restriction of the Poissonian noise restricted to ${\mc M}_{\pi, E}^N$. 

We could also decide to conserve energy and not the volume by adding a suitable perturbation. The invariant states are then given by $\mu_{\beta,0}$, $\beta>0$. If $V$ is even, a simple Poissonian noise consists to change the sign of $\eta_x$  independently on each site $x$ at random exponential times. In this case one can prove, as in \cite{BO}, that the energy diffuses in the sense that the Green-Kubo formula converges to a well defined finite value. For a generic $V$ a Brownian noise with generator $S$ given by $S=\sum_{x\in \bb Z} K_x^2$ with $K_x = V' (\eta_{x+1}) \partial_{\eta_x} - V' (\eta_x) \partial_{\eta_{x+1}}$ makes the job.

Consider now the case where we want to add a stochastic perturbation conserving only the volume. It does not seem to be easy to define a simple Poissonian noise with such a property. A Brownian noise is obtained by the following scheme. Fix $\beta>0$, consider the vector field $X_x = \partial_{\eta_{x+1}} -\partial_{\eta_x}$ which is tangent to the hyperplane $\{ (\eta_x, \eta_{x+1}) \in \RR^2\; ; \; \eta_x + \eta_{x+1} =1\}$ and define the Langevin operator ${S}_{\beta}$ by
\begin{equation*}
\begin{split}
{S}_\beta &= \frac{1}{2} \sum_{x\in \bb Z} e^{-{\mc H}_{\beta, \lambda}} X_x (e^{{\mc H}_{\beta,\lambda}} X_x )\\
&= \frac{1}{2} \sum_{x\in \bb Z} X_x^2 + \frac{\beta}{2} \sum_{x\in{\ZZ}} \Big(V' (\eta_{x+1}) -V' (\eta_x)\Big) X_x
\end{split}
\end{equation*}
where ${\mc H}_{\beta, \lambda} = \beta \sum_{x\in \bb Z} V(\eta_x) +\lambda \sum_{x\in \bb Z} \eta_x$. Observe that ${S}_\beta$ depends on $\beta$ but is independent of $\lambda$. The operator ${S}_{\beta}$ is a nonpositive self-adjoint operator in ${\bb L}^{2} (\mu_{\beta, \lambda})$ for any $\lambda$ and ${S}_{\beta} (\sum_{x\in \bb Z} \eta_x) =0$. Then, the perturbed volume-conserving model has a generator $L^V_{\beta}$ given by
\begin{equation}
L^V_{\beta} = A + \gamma S_\beta
\end{equation}
where $\gamma>0$ is a parameter fixing the strength of the noise. By construction, the Markov process generated by $L^V_{\beta}$ has $\mu_{\beta, \lambda}$ as a set of invariant probability measures. In fact, using the same methods as in \cite{BS,FFL}  one can prove that the only space-time invariant probability measures with finite local entropy density are mixtures of the $(\mu_{\beta, \lambda})_{\lambda}$. We can also rewrite $L^V_{\beta}$ as
\begin{equation*}
\begin{split}
L^V_{\beta} & = \sum_{x\in \bb Z} \left\{ \left(1 -\frac{\gamma \beta}{2} \right) V' (\eta_{x+1}) + \gamma \beta V' (\eta_x) - \left(1+ \frac{\gamma \beta}{2} \right) V' (\eta_{x-1}) \right\} \partial_{\eta_x}\\
&+ \gamma \sum_{x\in \bb Z}( \partial_{\eta_x}^2 -\partial_{\eta_x, \eta_{x+1}}^2).
\end{split}
\end{equation*}

The microscopic flux $j_{x,x+1}$ associated to the volume conservation law is defined by
\begin{equation*}
L^V_{\beta} (\eta_x) = -\nabla j_{x-1,x}, \quad j_{x-1,x} = -\left( 1+\frac{\gamma \beta}{2}\right) V' (\eta_{x-1}) - \left( 1 - \frac{\gamma \beta}{2} \right) V' (\eta_x).
\end{equation*}
The semi-discrete directed polymer model considered in \cite{BC} is, up to an irrelevant scaling factor $2$, recovered by taking $V(\eta) =e^{-\eta}$, $\beta=1$ and $\gamma=2$ (see (3.7) in \cite{Sp3}).  In \cite{BC} the authors show that for a particular non stationary initial condition (``wedge''), by developing a very nice theory of Macdonald processes, the system belongs to the Kardar-Parisi-Zhang universality class (\cite{Sp3}). Unfortunately one can not use their results or their methods to derive a more precise picture for the model with exponential interactions considered in this paper. For other potentials $V$ the theory developed by Borodin and Corwin in \cite{BC} can not be adapted but it would be very interesting to see if one can relate the models generated by $L^V_\beta$ to the semi-discrete directed polymer and deduce some qualitative information from the latter. The use of the variational formulas considered in this paper could be the way.

\appendix
\section{Existence of the infinite dynamics}
\label{sec:exis}

In this section we prove existence of the infinite volume dynamics $(\xi(t))_{t \ge 0}$. We focus here on the process $\xi$ but the same proof can be carried for the process $\eta$ (or just define $\eta$ in terms of $\xi$ by $\eta_x (t) =- b^{-1} \log \xi_x (t) $, $x \in \ZZ$.  To simplify notations we will assume $b=1$.

Since the interaction coming from the deterministic part is  non-quadratic at infinity, proving the existence of the infinite dynamics is a non trivial task. Nevertheless nice sophisticated techniques  have been introduced by Dobrushin and Fritz  in \cite{DF}. Here, we follow closely the approach of \cite{F3} (see also \cite{F1,F2}) adapted to our case. By itself, the strategy of the proof of existence of solutions is standard: we consider finite subsystems  and prove compactness of this family by means of an a priori bound for a quantity ${\bar E}$ which plays the role of an energy density.  The obtention of this a priori bound is however non trivial and is the main step to get the existence of the dynamics. The aim of this appendix is to show how to get such an a priori bound. The a priori bound we derive here for the infinite dynamics is also valid for finite subsystems corresponding to a finite set $\Lambda \subset \ZZ$  with a bound which is independent of the size of $\Lambda$. This proves then that the finite subsystems form a compact family from which one can extract a subsequence converging to the infinite dynamics.

We have first to specify the space of allowed configurations $\Omega \subset (0,+\infty)^{\ZZ}$. For $x\in\ZZ$, let $g(x) = 1 + \log (1+|x|)$ and denote by $E(\xi, \mu,\sigma)$, $\xi\in(0,+\infty)^{\ZZ}$, $\mu \in \ZZ$, $\sigma>0$, the quantities
\begin{equation*}
\begin{split}
&E(\xi,\mu,\sigma)= \sum_{|x - \mu| \le \sigma}  (1 + 2 \xi_x - \log (\xi_x)),\\
&{\bar E} (\xi)= \sup_{\mu \in \ZZ} \, \sup_{ \sigma \ge g(\mu)} \, \sigma^{-1} E (\xi,\mu, \sigma).
\end{split}
\end{equation*}

The quantity ${\bar E}$ is called the logarithmic fluctuation of energy and the set $\Omega$ is defined as
 \begin{equation*}
\Omega:=\{\xi\in{(0,+\infty)^\ZZ}: {\bar E} (\xi) <+ \infty\}.
 \end{equation*}
The configuration space $\Omega$ is equipped with the product topology and with the associated Borel structure. It is easy to see that $ \nu_{\beta, \lambda} (\Omega) =1$ for any $\beta>0$ and $\lambda>-1$.

Let ${\bf N} (t) = \{N_{x,x+1} (t) \, ; \, x \in \ZZ\}$ be a collection of independent Poison processes of intensity $\gamma>0$. The equations of motion corresponding to the generator ${\mc L}$ read as
\begin{equation}
\label{eq:app-stoch-eq}
{d \xi_x} =  \xi_x (\xi_{x+1} -\xi_{x-1}) dt + \nabla \left( (\xi_{x} -\xi_{x-1}) dN_{x-1,x} (t)\right),\quad  x \in \ZZ.
\end{equation}

Let $D(\RR_+, \RR)$ denote the space of c\`adl\`ag functions of $\RR_+$ into $\RR$ with the Skorohod topology and let ${\bb D}= [D(\RR_+, \RR)]^{\ZZ}$ equipped with the product topology and the associated Borel field ${\mc B}$. The smallest $\sigma$-algebra on which all projectionsrestricted to the time interval $[0,t]$ are measurable will be denoted by ${\mc B}_t$. Finally, suppose that we are given a probability measure ${\bf P}$ on ${\mc B}$ such that our Poisson processes $N_{x,x+1}$ are realized as components of the random element of ${\bb D}$.

\begin{definition}
A ${\mc B}_t$-adapted mapping $\xi (t):= \xi(t, {\bf N})$ of ${\bb D}$ into itself is called a tempered solution of (\ref{eq:app-stoch-eq}) with initial configuration $\xi^0 \in \Omega$ if $\xi (0) =\xi^0$, almost each trajectory $\xi (\cdot, {\bf N})$ satisfies the integral form of (\ref{eq:app-stoch-eq}), and the logarithmic energy fluctuation ${\bar E} (\xi (t))$ is bounded on finite intervals of time with probability one.
\end{definition}

\begin{theo}
\label{th:existdyn2}
For any $\xi^0 \in \Omega$, there exists a unique tempered solution of (\ref{eq:app-stoch-eq}) with initial configuration $\xi^0 \in \Omega$.
\end{theo}

As explained above, the main step to prove this theorem is to obtain an a priori bound that we prove in Proposition \ref{prop:apb}. For a complete proof, we refer to \cite{F3} ( or \cite{F1,F2}).\\

Now we notice that the Gibbs state $\nu_{\beta, \lambda}$, $(\beta, \lambda) \in (0, +\infty) \times (-1, +\infty)$ is formally invariant for the infinite dynamics generated by $(\xi (t))_{t \ge 0}$. This  can be seen by observing that $ \int ({\mc L} f)(\xi) d\nu_{\beta, \lambda}(d\xi)=0$ for nice functions $f:\Omega\rightarrow{\bb R}$. Nevertheless, some care has to be taken to prove this. Indeed, we do not know that ${\mc L}$ is really the generator of the semigroup generated by $(\xi (t))_{ t\ge 0}$ on the space of bounded measurable functions on $\Omega$ in the usual Hille-Yosida theory. This can be a very difficult question that we prefer to avoid (see \cite{F2}). Instead we use the fact that the infinite dynamics can be approximated by finite subsystems.

\begin{prop}
For any $\beta>0, \lambda > -1$, the probability measure $\nu_{\beta, \lambda}$ is invariant for the process $(\xi (t))_{t \ge 0}$.
\end{prop}

\begin{proof}
Let $n \ge 2$ and consider the local dynamics generated by the generator ${\mc L}_n = {\mc A}_n + \gamma {\mc S}_n $ where
\begin{equation*}
\begin{split}
({\mc A}_n f)(\xi) & = \sum_{x=-n}^n \xi_x (\xi_{x+1} -{\xi}_{x-1}) \partial_{\xi_x} f (\xi)\\
& \, -\, \xi_{n+1} \left( \xi_n +\frac{\lambda +1}{\beta}\right) \partial_{\xi_n} f(\xi) + \xi_{-n -1} \left( \xi_{-n} \, +\,  \frac{\lambda +1}{\beta} \right) \partial_{\xi_{-n-1}} f (\xi),\\
({\mc S}_n f)(\xi) &= \sum_{x=-n}^{n} \Big( f(\xi^{x,x+1}) -f(\xi)\Big)
\end{split}
\end{equation*}
where $f: \Omega \to \RR$ is a compactly supported continuously differentiable function. The dynamics is essentially finite-dimensional since the particles outside the box $\{-n-1, \ldots, n+1\}$ are frozen. Thus, the classical Hille-Yosida theory can be applied. The boundary conditions have been chosen to have
\begin{equation*}
 \int ({\mc L}_n f) (\xi) d \nu_{\beta, \lambda}(\xi) =0
\end{equation*}
for any compactly supported continuously differentiable function $f$ which shows that $\nu_{\beta, \lambda}$ is invariant for the local dynamics. Since, as a consequence of the a priori bound, the infinite dynamics is obtained as a limit of finite local dynamics, this implies that $\nu_{\beta, \lambda}$ is invariant for the infinite dynamics.
\end{proof}

Then this defines a strongly continuous semigroup of contractions $(P_t)_{t \ge 0}$ on the Hilbert space ${\bb L}^2 (\Omega, {\mc B}, \nu_{\beta, \lambda})$. Moreover, It\^o's formula shows that its generator is a closable extension of ${\mc L}$ given by ${\mc A} +\gamma {\mc S}$ since for any local compactly supported continuously differentiable function $f$, we have
\begin{equation*}
(P_t f) (\xi) = f (\xi) + \int_0^t ({P_s} {\mc L} f) (\xi) ds, \quad \xi \in \Omega, \quad t \ge 0.
\end{equation*}

\subsection{Logarithmic energy fluctuation}

We have first to consider a clever smooth modification of ${\bar E}$. Let $0<\lambda<1$ and consider a twice continuously differentiable nonincreasing function $\varphi: \RR \to (0,1)$ such that $\varphi (u) =e^{\lambda (1-u)}$ if $u \ge 2$, $\varphi (u) =(1+\lambda + \lambda^2 /2) e^{-\lambda}$ if $u \le 1$, and $\varphi$ is concave for $u \le 3/2$, convex if $u\ge 3/2$. Finally, $0 \le - \varphi' (u) \le \lambda \varphi (u) \le e^{\lambda(1-u)}$, $\varphi (u) \ge e^{-\lambda (1+u)}$ and $|\varphi^{\prime \prime} (u)| \le \varphi (u)$ for all $u>0$.

For $x\in\ZZ$ and $\sigma \ge 1$ we define the function $f$ as
\begin{equation*}
f(x,\sigma)= \int_{\RR} {\varphi} (|x-y| /\sigma) e^{-2\lambda|y|} dy.
\end{equation*}

In \cite{F2} are proved the following properties on $f$:
\begin{eqnarray}
\begin{split}
\label{eq:equff1} &c_1 \exp (-\lambda |x| / \sigma) \le f(x,\sigma) \le c_2 \exp (-\lambda |x| /\sigma), \\
 \label{eq:equff2} &f(x,\sigma) \le f(y,\sigma) e^{2 \lambda |x-y|}, \quad \partial_ \sigma f (x,\sigma) \le e^{2 \lambda | x-y|} \partial_{\sigma} f (y,\sigma). \\
 \label{eq:equff3} &| \partial_x f (x,\sigma) | \le \min \{\partial_\sigma f  (x,\sigma) , \sigma^{-1} f (x,\sigma)\},\\
 \label{eq:equff4}&g(x) |\partial_x f (x-\mu, \sigma) | \le 4 g(|\mu| +\sigma)\, ( \partial_{\sigma} f)(x-\mu, \sigma).
\end{split}
\end{eqnarray}

Here the constants depend only on $\lambda$.

For $\xi\in(0,+\infty)^{\ZZ}$, $\mu\in{\ZZ}$ and $\sigma>0$, consider the function
\begin{equation}
W(\xi, \mu, \sigma) =\sum_{x \in \ZZ} f(x-\mu, \sigma) (1+ 2 \xi_x - \log \xi_x)
\end{equation}
and let
\begin{equation}
{\bar W} (\xi) = \sup_{\mu \in \ZZ} \, \sup_{\sigma \ge g(\mu)} \,\Big\{ \sigma^{-1} W (\xi,\mu,\sigma)\Big\}.
\end{equation}

Observe that by (\ref{eq:equff1}),
\begin{equation}
\label{eq:compWE}
W(\xi,\mu,\sigma) \ge c_1e^{-\lambda} E(\xi, \mu,\sigma),
\end{equation}
for all $\xi\in(0,+\infty)^{\ZZ}$, $\mu\in{\ZZ}$ and $\sigma>0$.

For  $\xi\in(0,+\infty)^{\ZZ}$, we also consider the function
\begin{equation}
{\widehat W} (\xi) = \sup_{\mu \in \ZZ} \Big\{\frac{W(\xi, \mu,g(\mu))}{g(\mu)}\Big\}.
\end{equation}

The following lemma shows that these two modifications of the logarithmic energy fluctuation are equivalent to ${\bar E}$.

\begin{lemma}
\label{lem: comp6}

There exists a constant $C$ such that for all $\xi\in(0,+\infty)^{\ZZ}$:
\begin{equation*}
C^{-1} {\widehat W} (\xi) \le {\bar W}(\xi)\le C {\widehat W}(\xi), \quad C^{-1} {\bar E} (\xi)\le {\bar W}(\xi)  \le C {\bar E}(\xi).
\end{equation*}
\end{lemma}

\begin{proof}
The inequality ${\widehat W}(\xi)  \le {\bar W}(\xi) $  for all $\xi\in(0,+\infty)^{\ZZ}$, is trivial. Let us prove the second one by taking $\sigma \ge g(\mu)$, $\mu \in \ZZ$ and denoting $1+ 2\xi_x -\log \xi_x$ by $H_x$. By  (\ref{eq:equff1}), we have
\begin{equation*}
\begin{split}
W(\xi, \mu,\sigma)& \le c_2 \sum_{x \in \ZZ} \exp \left( -\lambda |x-\mu|/\sigma \right) H_x =c_2 \sum_{n =0}^{\infty} e^{- \lambda n/\sigma}  \sum_{|x-\mu|=n} H_x\\
&=c_2 (1- e^{- \lambda  /\sigma}) \sum_{n=0}^{\infty} e^{- \lambda n /\sigma}  \sum_{|x-\mu| \le n} H_x,
\end{split}
\end{equation*}
where the last equality follows from $\sum_{|x-\mu|=n} H_x =  \sum_{|x-\mu| \le n} H_x  - \sum_{|x-\mu| \le n-1} H_x $ and a discrete  integration by parts. Let $r\ge 1$ be the integer such that $r-1 < g(\mu) \le r$ and decompose the set $\{x \in \ZZ \, ;\, |x-\mu| \le n\}$ as $\cup_{j=1}^{K+1} \Lambda_j $ where the $\Lambda_j$ are non intersecting intervals of length $r$ for $j=1, \ldots,K$ and $\Lambda_{K+1}$ is of length at most $r-1$. Observe that $K+1$ is of order $n/ g(\mu)$. By using (\ref{eq:compWE}), we have easily that
\begin{equation*}
\sum_{x \in \Lambda_j} H_x \le C \, g(\mu)\,  {\widehat W} (\xi)
\end{equation*}
where $C$ depends only on $\lambda$. Thus we get
\begin{equation*}
\begin{split}
W(\xi, \mu,\sigma)& \le C  (1- e^{- \lambda  /\sigma}) \sum_{n=0}^{\infty} e^{- \lambda n /\sigma} n  {\widehat W} (\xi)\le C' \sigma {\widehat W} (\xi)
\end{split}
\end{equation*}
which concludes the proof of the second inequality.

The proof of $C^{-1} {\bar E}(\xi) \le {\bar W}(\xi)  \le C {\bar E}(\xi) $ for all $\xi\in(0,+\infty)^{\ZZ}$, is the same. The first inequality follows from \eqref{eq:compWE}
and the constant can be taken equal to $c_1e^{-\lambda}$. The second inequality follows from a similar argument to the one used above.
\end{proof}

\subsection{The {\textit{a priori}} bound}

\begin{prop}[A priori bound]
\label{prop:apb}
For each $w \ge 1$ there exists a continuous function $q_{w} (t)$, $t\ge 0$, such that
\begin{equation*}
{\bf P} \left\{ \sup_{0 \le s \le t} {\bar W} (\xi(s)) > \exp(q_{w} (t) g(u)) \right\} \le e^{-u}
\end{equation*}
for each $u \ge 1$, $t \ge 0$, whenever ${\bar W} (\xi^0) \le w$ and $(\xi (t))_{t \ge 0}$ is a tempered solution of (\ref{eq:app-stoch-eq}) with initial condition $\xi^0$.
\end{prop}

 \begin{proof}
We consider a tempered solution $(\xi (t))_{t \ge 0}$ of (\ref{eq:app-stoch-eq}) with initial configuration $\xi^0 \in \Omega$.

For each $k \ge 1$, $\mu \in \ZZ$ and $t\geq{0}$ we define the stochastic process $\rho_{k}$ by
\begin{equation}
\rho_k (t) = k g (\mu) - C_0 \int_0^t g (|\mu| +|\rho_k (s)| ) Z' (s)  ds
\end{equation}
where $C_0:=C_0(\gamma, \lambda)$ is a positive constant that will be chosen later and
\begin{equation*}
Z(t)= \int_0^t {\bar W} (\xi (s)) ds.
\end{equation*}

Since the function $f(\cdot)$ is positive, $\bar W(\cdot)$ is also positive and this turns $Z(\cdot)$ positive. The trajectories of $\rho_k$ are differentiable, decreasing and satisfy $\rho_{k+1} (t) - \rho_k (t) \le g(\mu)$ a.s. for each $t\ge 0$. We consider also the sequence of stopping times $\tau_k =\inf \{ t \ge 0 \, ; \, \rho_k (t) \le g(\mu) \}$ which satisfy $\tau_k < \tau_{k+1} < +\infty$ and $\lim_{k \to  \infty} \tau_k =  \infty$ a.s.
We evaluate now the stochastic differential of $t \to W (\xi (t), \mu, \rho_k (t))$ for $t \le \tau_k$ (so that $\rho_k (t) \ge 1$). This is given by
\begin{equation*}
\begin{split}
d \left[  W (\xi (t), \mu, \rho_k (t)) \right] &=   I^{(k)}_0 (t) dt -   C_0   (\partial_\sigma W) (\xi(t), \mu, \rho_k (t)) g (|\mu| +\rho_k (t)) {\bar W} (\xi (t)) dt\\
&  + \;  dI^{(k)}_1 (t)
\end{split}
\end{equation*}
where
\begin{equation}
\label{eq:I0k}
\begin{split}
I_0^{(k)} (t) =&2 \sum_{x \in \ZZ} \Big(  f(x- \mu, \rho_k (t) ) - f(x+1- \mu, \rho_k (t) )\Big) \xi_x (t) \xi_{x+1} (t) \\
+& \sum_{x \in \ZZ} \Big(f(x+1- \mu, \rho_k (t) ) - f (x-1-\mu, \rho_k (t) ) \Big) \xi_x (t)
\end{split}
\end{equation}
and
\begin{equation*}
d I^{(k)}_1 =\sum_{x\in\ZZ} f(x- \mu, \rho_k ) \Big\{ 2 \nabla\Big((\xi_{x}  -\xi_{x-1} ) dN_{x-1,x} \Big) -  \nabla\Big( (\log \xi_{x} -\log \xi_{x-1}) dN_{x-1,x} \Big)  \Big\}.
\end{equation*}

We first estimate the term $I_0^{(k)} (t)$ and we show that if $C_0$ is taken sufficiently large then, for $t \le \tau_k$ we have that
\begin{equation}
\label{eq:I0kne}
I_{0}^{(k)} (t) \; - \;  C_0 \,   (\partial_\sigma W) (\xi(t), \mu, \rho_k (t)) \, g (|\mu| +\rho_k (t)) {\bar W} (\xi (t)) \le 0.
\end{equation}

The second term on the right hand side of (\ref{eq:I0k}) can be estimated, by using (\ref{eq:equff2}) and (\ref{eq:equff3}),  to get to

\begin{equation}\label{detailed estimate}
\begin{split}
\Big| f(x+1- \mu, \rho_k (t) ) - f (x-1-\mu, \rho_k (t) ) \Big| &= \Big| \int_{-1}^1 \,  (\partial_x f) (x-\mu +\alpha, \rho_k (t)) d\alpha \Big|\\
& \le  \int_{-1}^1  \,   \Big| (\partial_x f) (x-\mu +\alpha, \rho_k (t))d\alpha \Big|\\
& \le  \int_{-1}^1 \,   (\partial_\sigma f) (x-\mu +\alpha, \rho_k (t))d\alpha \\
& \le 2 \sup_{[ x- \mu -1, x -\mu +1]}\Big\{ \partial_{\sigma} f( \cdot, \rho_{k} (t))\Big\} \\
&\le 2 e^{2 \lambda}  \partial_{\sigma} f( x-\mu, \rho_{k} (t))
\end{split}
\end{equation}
 which gives us that
\begin{equation*}
\begin{split}
\sum_{x \in \ZZ} \Big(f(x+1- \mu, \rho_k (t) ) - f (x-1-\mu, \rho_k (t) ) \Big) \xi_x (t) &\le C \sum_{x \in \ZZ} \partial_\sigma f (x-\mu, \rho_k (t) )  \xi_x (t)\\
 &\leq C (\partial_\sigma W) (\xi(t), \mu, \rho_k (t)).
 \end{split}
\end{equation*}

Now, notice that for any $x \in \ZZ$ and for all $\xi\in{(0,\infty)^{\ZZ}}$ we have that
\begin{equation*}
 \bar W(\xi)\geq{\hat W(\xi)}\geq{\frac{W(\xi,x,g(x))}{g(x)}}.
 \end{equation*}
On the other hand, by \eqref{eq:compWE} and since for all $x>0$ it holds that $\log(x)\leq{1+x}$, then we have that $W(\xi,x,g(x))\geq{c_1e^{\lambda}E(\xi,x,g(x))}\geq{c_1e^{-\lambda}\xi_{x+1}}$.
Then, we conclude that there exists a constant $C$ such that for all $x\in{\ZZ}$ and $\xi\in{(0,\infty)}^{\ZZ}$,
 \begin{equation}\label{useful estimate}
 \xi_{x+1} \le Cg(x) {\bar W} (\xi).
 \end{equation}

To estimate the first term on the right hand side of (\ref{eq:I0k}) we use the previous estimate, (\ref{eq:equff4}) and a similar argument as done in \eqref{detailed estimate}. It follows that
\begin{equation*}
\begin{split}
&\Big|\sum_{x \in \ZZ} \left( f(x- \mu, \rho_k (t) ) - f(x+1- \mu, \rho_k (t) )\right) \xi_x (t) \xi_{x+1} (t)\Big| \\
& \le C \, {\bar W} (\xi (t)) \, g(|\mu| +\rho_k (t) )  \sum_{x \in \ZZ} \partial_\sigma f (x-\mu, \rho_k (t) ) \xi_x (t)\\
&\leq C \, {\bar W} (\xi (t)) \, g(|\mu| +\rho_k (t) ) \, (\partial_\sigma W) (\xi(t), \mu, \rho_k (t)).
\end{split}
\end{equation*}

Then, (\ref{eq:I0kne}) follows.

The  term $dI^{(k)}_1$ can be written as
\begin{equation*}
\begin{split}
dI^{(k)}_1 \!\!&=\sum_{x \in \ZZ}\!f(x-\mu,\rho_k )\!\Big\{ 2 \nabla\left( (\xi_{x} -\xi_{x-1}) dN_{x-1,x} \right)\!-\!\nabla( (\log \xi_{x} -\log \xi_{x-1}) dN_{x-1,x})\Big\}\\
&= - \sum_{x \in \ZZ} \left( f(x+1-\mu,\rho_k ) -f(x- \mu, \rho_k ) \right) \left\{ 2 \nabla \xi_x -\nabla \log \xi_x \right\} dN_{x,x+1}\\
&=  - \sum_{x \in \ZZ} \left( f(x+1-\mu,\rho_k ) -f(x- \mu, \rho_k) \right) \left\{ 2 \nabla \xi_x -\nabla \log \xi_x \right\} (dN_{x,x+1}  -\gamma dt)\\
&- \gamma \sum_{x \in \ZZ} \left( f(x+1-\mu,\rho_k ) -f(x- \mu, \rho_k ) \right) \left\{ 2 \nabla \xi_x -\nabla \log \xi_x \right\} dt.
\end{split}
\end{equation*}

Since the compensated Poisson processes $N_{x,x+1} (t) -\gamma t$ are orthogonal martingales with quadratic variation $\gamma^2 t$, then
$$dM^{(k)}_{\mu}\! =\! -\! \sum_{x \in \ZZ} \left( f(x+1-\mu,\rho_k) \!-\!f(x- \mu, \rho_k ) \right) \Big\{ 2 \nabla \xi_x\! -\!\nabla \log \xi_x \Big\} (dN_{x,x+1} -\gamma dt)$$
defines a martingale with a quadratic variation equal to
 $$d\langle M^{(k)}_{\mu} \rangle_t =\gamma^2 \sum_{x \in \ZZ} \left( f(x+1-\mu,\rho_k ) -f(x- \mu, \rho_k ) \right)^2 \Big\{ 2 \nabla \xi_x -\nabla \log \xi_x \Big\}^2 \, dt. $$

Using a similar argument to the one in \eqref{detailed estimate}, together with the fact that for all $x,y\in \ZZ$ such that $|x|, |y| \le C$ it holds that $|x-y|^2\leq 2C|x-y|$, the boundedness of the function $f$, (\ref{eq:equff2}), (\ref{eq:equff3}), (\ref{eq:equff4}) and \eqref{useful estimate}, one has that there exists a constant $C$ such that
\begin{equation*}
d \langle M^{(k)}_\mu \rangle_t \le C\;  g (|\mu| + \rho_k (t) ) \; {\bar W} (\xi (t)) \; \partial_{\sigma} W(\xi (t), \mu, \rho_k (t)) \, dt.
\end{equation*}

Similarly we obtain that
\begin{equation*}
\begin{split}
&\left|  \sum_{x \in \ZZ} \left[ f(x+1-\mu,\rho_k (t)) -f(x- \mu, \rho_k (t)) \right] \left\{ 2 \nabla \xi_x (t) -\nabla \log \xi_x (t) \right\} \right| \\
&\quad \quad \quad \quad \quad \quad \quad \quad \quad \quad \quad \quad \quad \quad \quad \quad \quad \quad \quad \quad \quad \quad \quad \quad \le C \,  \partial_{\sigma} W (\xi (t), \mu, \rho_k (t)).
\end{split}
\end{equation*}

Thus, if the constant $C_0$ is chosen sufficiently  large, we have
\begin{equation*}
\begin{split}
\sup_{t \ge 0} \Big\{W (\xi (t\wedge \tau_k ), \mu, \rho_k (t \wedge \tau_k )) \Big\}\le W (\xi (0), \mu, k g(\mu))  + \sup_{t \ge 0}\Big\{ N(\mu,k, t)\Big\}
 \end{split}
\end{equation*}
where $N(\mu,k, t)= M^{(k)}_{\mu} (t \wedge \tau_k)- \frac{1}{2} \langle M^{(k)}_{\mu} \rangle_{t \wedge \tau_k}$. Observe that $\exp ( M^{(k)}_{\mu} (t \wedge \tau_k) -\frac12 \langle M^{(k)}_{\mu}  \rangle_{t \wedge \tau_k})$ is a martingale with expectation equal to $1$. By the exponential supermartingale  inequality, we have that
$${\bf P } \Big(\sup_{t \ge 0}\Big\{ N (\mu, k, t) > u\Big\} \Big)  \le e^{-u}.$$
Thus we proved that for each $k \ge 1$, $\mu \in \ZZ$ and $u>0$,
\begin{equation}
\label{eq:diavolo}
\sup_{t \ge 0} \Big\{W (\xi (t \wedge \tau_k),\mu, \rho_k (t \wedge \tau_k))\Big\} \le W ( \xi (0), \mu, kg(\mu)) +u
\end{equation}
with a probability greater than $1-e^{-u}$. Applying (\ref{eq:diavolo}) for each $\mu \in \ZZ$ and $k \ge 1$ with $u$ replaced by $u +Ak g(\mu)$ where $A\ge 1$ is sufficiently large to have $\sum_{k \ge 1} \sum_{\mu \in \ZZ} e^{-Ak g(\mu)} \le 1$, we obtain
\begin{equation}
\label{eq:diavolo2}
\begin{split}
\sup_{t \ge 0}\Big\{ W (\xi (t \wedge \tau_k),\mu, \rho_k (t \wedge \tau_k))\Big\} &\le W ( \xi (0), \mu, kg(\mu)) +Ak g(\mu) +u \\
& \le k g(\mu) {\bar W} (\xi (0)) + A k g(\mu) +u
\end{split}
\end{equation}
with a probability greater than $1-e^{-u}$ uniformly in $k$ and $\mu$.

Define now $k:=k_t$, $t \ge 0$ as the smallest integer $k \ge 1$ for which $\rho_k (t) > g(\mu)$; then $\tau_k >t$ and $\rho_k (t) \le 2 g(\mu)$ as $\rho_{k-1} (t) \le g(\mu)$; thus choosing $k=k_t$ in (\ref{eq:diavolo2}) and using that $W(\xi,\mu,\sigma)$ is increasing in $\sigma$ (since $\partial_{\sigma} f \ge 0$ by the conditions imposed on $\varphi$), we get
\begin{equation*}
\frac{W(\xi (t), \mu,g(\mu))}{g(\mu)} \!\le\! \frac{W(\xi (t), \mu, \rho_k (t) )}{g(\mu)} \!\le\! k {\bar W} (\xi (0))\! +Ak\! +\!\frac{u}{g(\mu)}\!\leq{k {\bar W} (\xi (0)) \!+Ak \!+{u}},
\end{equation*}
 where in the last inequality we used the fact that $g(x)\geq{1}$ for all $x\in{\RR}$.
Taking the supremum over $\mu$ and using Lemma \ref{lem: comp6}, we obtain
\begin{equation*}
{\bar W} (\xi (t)) \le C k_t {\bar W} (\xi (0)) +u
\end{equation*}
for each $t\ge 0$ with probability at least $1-e^{-u}$. On the other hand,
\begin{equation*}
2 g(\mu) \ge k_t g(\mu) - C_0 \int_0^t g (|\mu| + |\rho_k (s)|) Z' (s) ds
\end{equation*}
whence
\begin{equation*}
k_t \le 2 + C_0 \int_0^t \frac{g (|\mu| + |\rho_k (s)|)}{g(\mu)} Z' (s) ds.
\end{equation*}
Since $\rho_k (s) \le k_t g(\mu)$ for any $s \in [0,t]$  and $g$ is increasing, we have that $g (|\mu| + |\rho_k (s)|)\leq{g(\mu+k_tg(\mu))}$. On the other hand for $x \ge 2$, $g(x) \le x$ together with the fact that for $x,y\in \RR$ $g(|x||y|) \le g(|x|) g(|y|)$ and since $g(1+x) \le 1 +g(x)$ for $x \ge 1$, we obtain that $g(\mu+k_tg(\mu))\leq{g(\mu)(1+g(k_t))}$. As a consequence we obtain that
\begin{equation}\label{final estimate}
k_t\leq{2+C_0Z(t) (1+ g(k_t))}.
\end{equation}
Since for all $x\geq{1}$ we have that  $g(x) \le 1+2 \sqrt{|x|}$, then

\begin{equation*}
k_t \le 2 + C_0 Z(t) (2 +2 {\sqrt{k_t}}).
\end{equation*}
Finally, it follows that ${\sqrt{k_t}} \le 2 + 4C_0 Z(t)$. Then, since $g$ is increasing and by plugging the previous  inequality in \eqref{final estimate}, we obtain that
\begin{equation*}
k_t \le 2 + C_0 Z (t) (1 +g((2+4C_0 Z(t))^2) ).
\end{equation*}
Recalling that $Z'(t)={\bar W} (\xi (t))$ we obtained that there exists a constant $M >0$ depending only on $\lambda$ such that for any $w  \ge 1$ and any initial condition $\xi (0)$ satisfying ${\bar W} (\xi (0)) \le w$,
\begin{equation*}
{\bf P} \left[ \sup_{t \ge 0} \Big\{M^{-1} Z' (t) - w (1 + Z(t) g(Z(t)))\Big\} \le u \right] \ge 1-e^{-u}.
\end{equation*}

The a priori bound follows from this last inequality (see \cite{F3}, Proposition 1).
\end{proof}

\section*{Acknowledgements}
The authors are very grateful to J\'ozsef Fritz for illuminating discussions on the existence of the infinite dynamics.  We acknowledge the support of the French Ministry of
Education through the grant ANR-10-BLAN 0108 (SHEPI). We are grateful to \' Egide and FCT for the research project FCT/1560/25/1/2012/S.
 We are grateful to FCT (Portugal) for support through the research project PTDC/MAT/109844/2009. PG thanks the Research Centre of Mathematics of the University of Minho, for the financial support provided by "FEDER" through the "Programa Operacional Factores de Competitividade – COMPETE" and by FCT through the research project PEst-C/MAT/UI0013/2011. PG thanks the warm hospitality of ``Courant Institute of Mathematical Sciences", where part of this work was done.


\end{document}